\documentclass[11pt]{article}

\usepackage{amssymb}
\usepackage{amsfonts}
\usepackage{amsmath}
\usepackage{amsthm}
\usepackage[english]{babel}
\usepackage{latexsym}
\usepackage{color}
\usepackage{enumerate}
\usepackage{enumitem}
\usepackage{nicefrac}
\usepackage{mathrsfs}
\usepackage{comment}
\usepackage[hmargin=2cm,vmargin=2cm]{geometry}
\usepackage{bm}
\usepackage{bbm}
\usepackage{tikz}
\usepackage{cite} 

\usepackage{centernot}

\usepackage[affil-it]{authblk}
\theoremstyle{plain}
\newtheorem{theorem}{Theorem}[section]
\newtheorem{lemma}[theorem]{Lemma}
\newtheorem{example}[theorem]{Example}
\newtheorem{proposition}[theorem]{Proposition}
\newtheorem{corollary}[theorem]{Corollary}

\theoremstyle{definition}
\newtheorem{definition}[theorem]{Definition}
\newtheorem{assumption}[theorem]{Assumption}
\newtheorem*{assumption*}{Standing Assumption}
\newtheorem{remark}[theorem]{Remark}

\theoremstyle{remark}

\numberwithin{equation}{section}

\newcommand{\ba}{\begin{array}{ll}}
\newcommand{\bal}{\begin{array}{ll}}
\newcommand{\ea}{\end{array}}

\newcommand{\E}{\mathbb{E}}

\newcommand{\probp}{\mathbb{P}}
\newcommand{\probq}{\mathbb{Q}}
\newcommand{\R}{\mathbb{R}}
\newcommand{\N}{\mathbb{N}}

\newcommand{\cN}{{\mathcal{N}}}
\newcommand{\cF}{{\mathcal{F}}}

\newcommand{\cB}{{\mathcal{B}}}
\newcommand{\cS}{{\mathcal{S}}}
\newcommand{\cA}{\mathcal{A}}
\newcommand{\cC}{\mathcal{C}}
\newcommand{\cD}{\mathcal{D}}

\newcommand{\cK}{\mathcal{K}}

\newcommand{\cL}{\mathcal{L}}

\newcommand{\cP}{\mathcal{P}}

\newcommand{\cU}{\mathcal{U}}
\newcommand{\cV}{\mathcal{V}}

\newcommand{\cX}{{\mathcal{X}}}
\newcommand{\cY}{{\mathcal{Y}}}

\newcommand{\cone}{\mathop{\rm cone}\nolimits}

\newcommand{\Span}{\mathop{\rm span}\nolimits}
\newcommand{\cl}{\mathop{\rm cl}\nolimits}

\newcommand{\norm}[1]{\left\| #1 \right\|}

\newcommand{\RVaR}{\mathop {\rm RVaR}\nolimits}

\newcommand{\ES}{\mathop {\rm ES}\nolimits}

\newcommand{\epi}{\mathop {\rm epi}\nolimits}

\def\keywords{\vspace{.5em}
{\noindent\textbf{Keywords}:\,\relax%
}}

\def\JELclassification{\vspace{.5em}
{\noindent\textbf{JEL classification}:\,\relax%
}}

\def\MSCclassification{\vspace{.5em}
{\noindent\textbf{MSC}:\,\relax%
}}

\makeatletter
\def\@fnsymbol#1{\ensuremath{\ifcase#1\or 1\or 2\or 3\or 4\or 5\or 6\or 7\or 8\else\@ctrerr\fi}}
\makeatother

\begin{document}

\title{Risk measures beyond frictionless markets}

\author{\sc{Maria Arduca} \thanks{Department of Economics and Finance, LUISS Guido Carli University, Rome, Italy, \texttt{marduca@luiss.it}}
\ \ \, \ \  \sc{Cosimo Munari} \thanks{Center for Finance and Insurance and Swiss Finance Institute, University of Zurich, Switzerland, \texttt{cosimo.munari@bf.uzh.ch}}}

\date{\today}

\maketitle

\begin{abstract}
We develop a general theory of risk measures that determines the optimal amount of capital to raise and invest in a portfolio of reference traded securities in order to meet a pre-specified regulatory requirement. The distinguishing feature of our approach is that we embed portfolio constraints and transaction costs into the securities market. As a consequence, we have to dispense with the property of translation invariance, which plays a key role in the classical theory. We provide a comprehensive analysis of relevant properties such as star shapedness, positive homogeneity, convexity, quasiconvexity, subadditivity, and lower semicontinuity. In addition, we establish dual representations for convex and quasiconvex risk measures. In the convex case, the absence of a special kind of arbitrage opportunities allows to obtain dual representations in terms of pricing rules that respect market bid-ask spreads and assign a strictly positive price to each nonzero position in the regulator's acceptance set.
\end{abstract}


\JELclassification{D81, G11, C61}

\MSCclassification{91G70, 46A20, 42A40}

\parindent 0em \noindent


\section{Introduction}

The publication of \cite{article:artzner1999coherent} has triggered an impressive amount of research on the topic {\em risk measures}, covering applications to portfolio risk management, capital allocation, pricing and hedging, risk sharing, and solvency regulation. Besides laying down the axioms for coherent risk measures, the main conceptual contribution of that paper is arguably to propose a way to construct general risk measures based on two financial primitives: acceptance sets and eligible assets. Let $\cX$ denote a suitable space of random variables representing financial variables of interest at a given future date. An acceptance set is any target set $\cA\subset\cX$ consisting of desirable variables. An eligible asset is any financial security that is used to reach acceptability. If the security is traded in a frictionless and liquid one-period market, as stipulated in the original paper, then it can be modelled as a couple $(S_0,S_1)\in\R\times\cX$, where $S_0$ and $S_1$ are its initial price and terminal payoff. The corresponding risk measure is defined by
\[
\rho(X) = \inf\{xS_0 \,; \ x\in\R, \ X+xS_1\in\cA\}, \ \ \ \ X\in\cX.
\]
The quantity $\rho(X)$ should be interpreted as the minimal amount of capital that has to be raised at the initial date and invested in the eligible asset to ensure acceptability of $X$. The bulk of the subsequent literature has focused on riskless eligible assets. We refer to \cite{book:follmer2011stochastic} for a comprehensive treatment of this case. The choice of a riskless eligible asset has often been motivated by a ``discounting'' argument; see, e.g., \cite{article:delbaen2002coherent}, \cite{article:follmer2002convex}, \cite{article:frittelli2002putting}. We refer to \cite{article:farkas2014beyond} and \cite{article:farkas2014capital} for a critical discussion about such argument and for a variety of results for general eligible assets. The preceding construction admits a natural extension to multiple eligible assets. If we consider $N$ securities traded in a frictionless and liquid market and model them by $(S^1_0,S^1_1),\dots,(S^N_0,S^N_1)\in\R\times\cX$, then the corresponding risk measure is given by
\[
\rho(X) = \inf\left\{\sum_{i=1}^Nx_iS^i_0 \,; \ x\in\R^N, \ X+\sum_{i=1}^Nx_i S^i_1\in\cA\right\}, \ \ \ \ X\in\cX.
\]
In this case, $\rho(X)$ is the minimal amount of capital that has to be raised at the initial date and invested in a portfolio of eligible assets to reach acceptability. This type of risk measures, which can be viewed as a generalization of superreplication prices, has been studied, e.g., in \cite{article:follmer2002convex}, \cite{article:artzner2009risk}, \cite{article:farkas2015measuring}, \cite{liebrich2019risk}, \cite{article:baes2020existence}.

\smallskip

In this paper, we dispense with the standard assumption that eligible assets are traded in a frictionless and liquid market and focus on risk measures of the form
\[
\rho(X) = \inf\{V_0(x) \,; \ x\in\cP, \ X+V_1(x)\in\cA\}, \ \ \ \ X\in\cX.
\]
The set $\cP\subset\R^N$ captures restrictions on the admissible portfolios of eligible assets, and the maps $V_0$ and $V_1$ assign to each portfolio of eligible assets its initial acquisition price and terminal liquidation value, respectively. The quantity $\rho(X)$ can therefore still be interpreted as the minimal amount of capital that has to be raised at the initial date and invested in a portfolio of eligible assets to reach acceptability. However, differently from the bulk of the literature, we stipulate no a priori assumption on the underlying financial market.
In particular, we allow for both proportional and nonproportional transaction costs and portfolio constraints. In addition, we do not posit that eligible assets pay off at the terminal date and we are thus forced to model terminal liquidation. Moreover, to capture important examples from practice, we consider both convex and nonconvex acceptance sets. By doing so, we aim to develop a general theory of risk measures beyond frictionless markets. This is a natural step to complement the existing literature. One of the main advantages of this general approach is that the standard properties of risk measures studied in the literature, like translation invariance, convexity, quasiconvexity, star shapedness, or (semi)continuity, can be interpreted as special cases of general properties expressed in terms of the underlying financial primitives and this helps fully appreciate their economic foundation.

\smallskip

From a technical perspective, a major challenge is the lack of translation invariance implied by market frictions. This property plays a crucial role in the standard theory. In particular, it greatly simplifies the study of lower semicontinuity, which is a necessary preliminary step to establish dual representations. Indeed, under translation invariance, lower semicontinuity is immediately implied by the closedness of the acceptance set in the single-asset case and by suitable no-arbitrage conditions in the multi-asset case. The failure of translation invariance requires pursuing a new strategy to establish lower semicontinuity of general risk measures (Theorem~\ref{thm:lsc_NoA_Llinear}). The key ingredient of our results is related to the absence of so-called scalable acceptable deals, which are special portfolios of eligible assets that, independently of their size, are admissible, can be acquired at zero cost, and deliver a nonzero acceptable liquidation payoff. We include a thorough discussion of scalable acceptable deals and their relationship with arbitrage opportunities (Propositions~\ref{prop: conditions on L} and~\ref{prop: sufficient conditions for NSAD}). Interestingly enough, in markets with frictions, there are situations where arbitrage opportunities arise whereas scalable acceptable deals cannot exist.

\smallskip

We establish dual representations for convex and quasiconvex risk measures providing a unifying perspective on dual representations in the frictionless literature. In the convex case, the domain of the dual representation is shown to consist of pricing rules defined on the entire model space that are consistent with market prices, i.e., lie in a suitably adjusted bid-ask spread, and respect the acceptance set, i.e., the range of prices assigned to acceptable positions is bounded from below (Theorem~\ref{theo: dual repr convex}). We also investigate when the domain can be further restricted to consist only of pricing rules that assign a strictly positive price to acceptable positions, and show that this special representation holds if the market admits no scalable acceptable deals (Theorem~\ref{theo: improved dual}). The corresponding dual representation can be seen as a risk-measure version of the classical Superhedging Theorem from arbitrage pricing. In the quasiconvex case, the domain is larger and, as illustrated by examples, must contain also pricing rules that are not consistent with market prices (Theorem~\ref{theo: dual representation quasiconvex}).

\smallskip

To the best of our knowledge, \cite{article:frittelli2006risk} is the only contribution to the risk measure literature where functionals similar to our general risk measures have been studied. The authors are concerned with reaching acceptability of streams of random variables through capital injections and, hence, they do not consider any market for eligible assets. However, as a preliminary mathematical step, they introduce functionals of the form
\[
\rho(X) = \inf\{\pi(Y) \,; \ Y\in\cC, \ X+Y\in\cA\}, \ \ \ \ X\in\cX.
\]
From our perspective, the elements of $\cC\subset\cX$ may be viewed as payoffs of portfolios of eligible assets and the map $\pi$ as a pricing rule. Our risk measures are at the same time less general, because $\cP$ is contained in a finite-dimensional space whereas $\cC$ need not be, and more general, because the presence of transaction costs at the terminal date, reflected by the lack of linearity of the liquidation map $V_1$, precludes us from attaching the same price to portfolios having the same liquidation value. Note that this is compatible with the absence of arbitrage opportunities in our market with frictions. In addition, the results of that paper are established under convexity assumptions and neither sufficient conditions for lower semicontinuity nor sufficient conditions to restrict the domain of dual representations to ``strict'' pricing rules are discussed there. Our work is also closely related to \cite{article:arduca2020market}. The focus of this paper is on good deal pricing in markets with frictions. In particular, the Superhedging Theorem established there resembles our dual representation with ``strict'' pricing rules. The main differences are that we work here with abstract model spaces and, once again, we have to model liquidation at the terminal date and do not work under convexity assumptions. This makes some key arguments used in that paper inapplicable in our case.

\smallskip

The paper is organized as follows. In Section~\ref{sec: model} we introduce the financial primitives and define general risk measures based on them. In Section~\ref{sect: properties} we derive from the properties of the primitives a variety of properties of general risk measures, featuring star shapedness, positive homogeneity, convexity, quasiconvexity, subadditivity. In Section \ref{sect: lsc} we establish sufficient conditions for lower semicontinuity under suitable extensions of the classical absence of arbitrage opportunities and under suitable regularity assumptions on the liquidation pricing rule. Section~\ref{sect: dual representations} is devoted to dual representations of convex and quasiconvex risk measures. A final appendix collects the necessary mathematical background.


\section{General risk measures}
\label{sec: model}

We consider an agent confronted with the problem of determining the optimal amount of capital to invest in an outstanding financial market in order to hedge a financial position at an acceptable level of risk. As is standard in risk measure theory, we select two reference dates. At the initial date $0$, the agent has to raise and invest capital. At the terminal date $1$, the agent's position materializes. The set of terminal financial positions is described by a real topological vector space $\cX$, which is assumed to be partially ordered by a convex cone $\cX_+$. For all $X,Y\in\cX$ we write $X\geq Y$ whenever $X-Y\in\cX_+$. The financial market consists of $N$ securities. Every vector in $\R^N$ is therefore interpreted as a portfolio of securities with the usual sign convention on long and short positions. For $i=1,\dots,N$ we denote by $e^i$ the $i$th unit vector in $\R^N$ corresponding to holding one unit of the $i$th asset. Four are the primitives of our problem:
\begin{itemize}
  \item The set of admissible portfolios $\cP\subset\R^N$.
  \item The acquisition pricing rule $V_0:\R^N\to\R$.
  \item The liquidation pricing rule $V_1:\R^N\to\cX$.
  \item The acceptance set $\cA\subset\cX$.
\end{itemize}
The set of admissible portfolios $\cP$ models portfolio constraints to which the agent may be subject, e.g., short selling restrictions or restricted access to certain asset classes or market segments. The acquisition pricing rule $V_0$ models ask prices at the initial date. Similarly, the liquidation pricing rule $V_1$ models bid prices at the terminal date. More precisely, for every portfolio $x\in\R^N$, the quantity $V_0(x)$ represents the amount of capital that is needed to buy $x$ at the initial date whereas the quantity $V_1(x)$ represents the amount of capital that is received by selling $x$ at the terminal date. In line with this interpretation, the quantity $-V_0(-x)$ represents the amount of capital that is received by selling $x$ at the initial date whereas the quantity $-V_1(-x)$ represents the amount of capital that is needed to buy $x$ at the terminal date. Finally, the acceptance set $\cA$ consists of all terminal financial positions that are deemed acceptable by the agent. The specific choice can be based on the agent's preferences towards risk or on external criteria, e.g., imposed by a financial regulator. Throughout the paper we work under the following assumptions.

\begin{assumption}
\begin{enumerate}
  \item[(1)] $0\in\cP$. 
  \item[(2)] $V_0(0)=0$ and $V_0(x)\geq-V_0(-x)$ for every $x\in\R^N$.
  \item[(4)] $V_1(0)=0$ and $V_1(x)\leq-V_1(-x)$ for every $x\in\R^N$.
  \item[(4)] $0\in\cA$ and $\cA+\cX_+\subset\cA$.
\end{enumerate}
\end{assumption}
The first property allows the agent to stay away from the market but otherwise imposes no restriction on portfolio constraints. The second and third properties stipulate that the pricing rules command a nonnegative bid-ask spread for every portfolio and are compatible with general forms of transaction costs (both proportional and nonproportional). Note that, provided the zero portfolio has zero acquisition/liquidation value, these properties are automatically satisfied if $V_0$ is convex and $V_1$ is concave. The fourth property is a standard rationality assumption in risk measure theory requiring that any position dominating an acceptable position must itself be acceptable.

\begin{remark}
(i) Our framework is flexible and compatible with each of the following two situations.
\begin{itemize}
  \item The basic securities pay off at the terminal date. In this case, the agent does not have to liquidate his or her portfolio in the market but simply cashes in the corresponding payoff. The liquidation rule $V_1$ is thus given by a simple aggregation of the individual payoffs and is therefore linear. This is the standard setting considered in the risk measure literature.
  \item Some of the basic securities pays off after the terminal date. In this case, the agent has to liquidate his or her portfolio at the prevailing market conditions. Whether or not $V_1$ is linear will depend on market frictions at the terminal date.
\end{itemize}

\smallskip

(ii) In the literature on risk measures in frictionless markets it is standard to work with payoffs instead of portfolios. This is possible because one assumes, implicitly or explicitly, that the liquidation rule $V_1$ is injective, so that for different portfolios $x,y\in\R^N$ one always has $V_1(x)\neq V_1(y)$. In a frictionless setting, this is equivalent to assuming that none of the basic securities is redundant, i.e., the payoffs of the basic securities are linearly independent. In turn, this law of one price allows to unambiguously define a pricing rule at the initial date applied directly to payoffs: For every $x\in\R^N$ the payoff $V_1(x)$ is assigned the price $V_0(x)$. In our general setting there seems to be no compelling reason to require injectivity of $V_1$. In particular, it should be observed that lack of injectivity of $V_1$ does not imply existence of arbitrage opportunities in a market with frictions.
\end{remark}

As said in the introduction, the agent's problem is to determine the minimal amount of capital that has to be raised and invested in an admissible portfolio of basic securities in order to ensure the acceptability of his or her outstanding financial position.

\begin{definition}
The {\em risk measure} associated to $(\cA,\cP,V_0,V_1)$ is the map $\rho:\cX\to[-\infty,\infty]$ defined by
\[
\rho(X):=\inf\{V_0(x) \,; \ x\in\cP, \ X+V_1(x)\in\cA\}, \ \ \ \ X\in\cX.
\]
\end{definition}

We conclude this section by collecting examples of the primitive elements $(\cA,\cP,V_0,V_1)$.

\begin{example}
Portfolio constraints have been widely investigated in the pricing literature with emphasis on proportional constraints ($\cP$ is a convex cone) or nonproportional constraints ($\cP$ is convex); see, e.g., \cite{broadie1998optimal}, \cite{jouini1999viability}, \cite{pham1999fundamental}. We list some standard examples covering the case of no portfolio constraints, no short selling, caps on long and short positions, margin requirements, and collateral requirements. Note that, depending on $V_0$, the constraints may be convex or not. The constraints below can be easily adapted to be binding for selected securities only.
\begin{itemize}
    \item $\cP=\R^N$.
    \item $\cP=\{x\in\R^N \,; \ x_i\geq0, \ \forall i=1,\dots,N\}$.
    \item $\cP=\{x\in\R^N \,; \ \underline{x}_i\leq x_i\leq \overline{x}_i, \ \forall i=1,\dots,N\}$ where $-\infty\leq\underline{x}_i<\overline{x}_i\leq\infty$ for $i=1,\dots,N$.
    \item $\cP=\{x\in\R^N \,; \ V_0(x_ie^i)+\gamma_iV_0(x)\geq0, \ \forall i=1,\dots,N\}$ where $\gamma_i>0$ for $i=1,\dots,N$.
    \item $\cP=\{x\in\R^N \,; \ \gamma V_0(\max\{x,0\})+V_0(\min\{x,0\})\geq0\}$ where $\gamma>0$.
\end{itemize}
\end{example}

\begin{example}
\label{ex: pricing rules}
Pricing rules beyond frictionless markets have also been thoroughly studied in the literature; see, e.g., \cite{pham1999fundamental}, \cite{cetin2007modeling}, \cite{article:pennanen2011arbitrage}. A standard approach is to define prices by way of aggregation. For $i=1,\dots,N$, let $p^a_i,p^b_i:\R_+\to\R_+$ be nondecreasing functions determining the ask and bid price of any given number of units of security $i$, respectively. Assume that $p^a_i(0)=p^b_i(0)=0$ and that bid-ask spreads are nonnegative, i.e., $p^a_i(x)\geq p^b_i(x)$ for every $x\in\R_+$. The functional $V_0:\R^N\to\R$ defined by
\[
V_0(x):=\sum_{x_i\geq0}p^a_i(x_i)-\sum_{x_i<0}p^b_i(-x_i), \ \ \ \ x\in\R^N,
\]
is an acquisition pricing rule that satisfies our standing assumptions. Depending on $p^a_i$ and $p^b_i$, one can cover the case of a frictionless market ($V_0$ is linear), a market with proportional transaction costs ($V_0$ is sublinear), and a market with nonproportional transaction costs ($V_0$ is convex or star shaped).
\begin{itemize}
\item If $p^a_i$ and $p^b_i$ are positively homogeneous and $p^a_i=p^b_i$ for $i=1,\dots,N$, then $V_0$ is linear.
\item If $p^a_i$ and $p^b_i$ are positively homogeneous for $i=1,\dots,N$, then $V_0$ is sublinear.
\item If $p^a_i$ is convex, $p^b_i$ is concave, and they are right continuous at $0$ for $i=1,\dots,N$, then $V_0$ is convex.
\item If $p^a_i$ is star shaped and $p^b_i$ is anti-star shaped for every $i=1,\dots,N$, then $V_0$ is star shaped.
\end{itemize}
The liquidation pricing rule can be defined by following a similar aggregation approach. For $i=1,\dots,N$, let $\varphi^a_i,\varphi^b_i:\R_+\to\R_+$be nondecreasing functions satisfying $\varphi^a_i(0)=\varphi^b_i(0)=0$ and $\varphi^a_i(x)\geq \varphi^b_i(x)$ for every $x\in\R_+$. Moreover, take $S^a_1,\dots,S^a_N,S^b_1,\dots,S^b_N\in\cX_+$. The map $V_1:\R^N\to\cX$ defined by
\[
V_1(x):=\sum_{x_i\geq0}\varphi^b_i(x_i)S^b_i-\sum_{x_i<0}\varphi^a_i(-x_i)S^a_i, \ \ \ \ x\in\R^N,
\]
is a liquidation pricing rule that satisfies our standing assumptions. As above, one can cover the case of a frictionless market ($V_1$ is linear), a market with proportional transaction costs ($V_1$ is superlinear), and a market with nonproportional transaction costs ($V_1$ is concave or anti-star shaped).
\end{example}

\begin{example}
\label{ex: kabanov 1p}
We define pricing rules in a currency market based on the setting of \cite{article:kabanov1999hedging} and \cite{article:schachermayer2004fundamental}. Fix a complete probability space $(\Omega,\cF,\probp)$. For $t=0,1$ the $N\times N$ bid-ask matrix $\Pi_t=(\pi^{ij}_t)$ rules the exchange between currencies at time $t$, i.e., $\pi^{ij}_t$ is the cost of buying one unit of currency $j$ in terms of currency $i$ at time $t$. Note that $\Pi_0$ is deterministic whereas $\Pi_1$ is random. As in \cite{article:schachermayer2004fundamental}, we assume that $\pi^{ij}_t>0$ and $\pi^{ii}_t=1$, and that $\pi^{ij}_t\leq\pi^{ik}_t\pi^{kj}_t$ for all $i,j,k=1,\dots,N$. The convex cone of positions that can be liquidated into the null position (possibly throwing away money) at time $t$ is defined as
\[
K_t:=\cone\{e^i \,; \ i=1,\dots,N, \ \pi^{ij}_te^i-e^j, \ 1\leq i,j \leq N\}.
\]
Clearly, the corresponding set of portfolios available at zero cost at time $t$ is $-K_t$. We denote by $K^*_t$ the polar cone of $-K_t$, which can be equivalently written as
\[
K^*_t =\{ z\in\R^N_+ \,; \ \pi^{ij}_tz_i\geq z_j, \ 1\leq i,j\leq N\}.
\]
Pricing rules can be defined by using the first currency as the numeraire as, e.g., in \cite{bouchard2001option} or \cite{astic2007no}. At the initial date, acquisition values are given by the smallest amount of the first currency that can be exchanged for a given portfolio, i.e.
\[
V_0(x) := \inf\{m\in\R \,; \ me^1-x\in K_0\}, \ \ \ x\in\R^N.
\]
It is easy to see that $V_0$ is sublinear. At the terminal date, liquidation values correspond to the largest amount of the first currency that can be obtained in exchange of a given portfolio, i.e.
\[
V_1(x) := \sup \{m\in\R \, ; \ x-me^1\in K_1\}, \ \ \ x\in\R^N.
\]
Clearly, $V_1$ is superlinear. To study ``finiteness'' of $V_0$ and $V_1$, define the set $Z_t := \{z\in K^\ast_t \,; \ z_1=1\}$. By the Bipolar Theorem, for every $x\in\R^N$ we can rewrite $V_0$ and $V_1$ as
\[
V_0(x) = \inf\left\{m\in\R \,; \ mz_1-x\cdot z\geq0, \ z\in K^\ast_0\right\} = \inf\left\{m\in\R \,; \ m-x\cdot z\geq0, \ z\in Z_0\right\} = -\sigma_{Z_0}(-x),
\]
\[
V_1(x) = \sup\left\{m\in\R \,; \ x\cdot z-mz_1\geq0, \ z\in K^\ast_1\right\} = \sup\left\{m\in\R \,; \ x\cdot z-m\geq0, \ z\in Z_1\right\} = \sigma_{Z_1}(x),
\]
where we used the standard notation for the scalar product in $\R^N$. As $Z_t$ is a polytope contained in the set $\{1\}\times\left[1/\pi^{21}_1,\pi^{12}_1\right]\times\dots\times\left[1/\pi^{N1}_1,\pi^{1N}_1\right]$, it follows by combining Example 1.2, Theorem 2.25, and Proposition 2.5 in \cite{book:molchanov2005theory} that $V_1(x)$ is a well-defined random variable and for every $x\in\R^N$
\begin{equation}
\label{eq: V1 valued in X}
|V_t(x)| \leq |x_1| + \sum_{i=2}^N |x_i|\pi_t^{1i}.
\end{equation}
This shows that $V_0$ takes finite values. In addition, it allows to control the outcomes of $V_1$ as well. For instance, if $\cX$ is a Banach lattice, e.g., an Orlicz space, then $V_1$ takes values in $\cX$ whenever the ask prices $\pi^{1i}_1$ belong to $\cX$ for $i=2,\dots,N$. The same example can be cast in the general setting of \cite{article:kaval2006link}, where pricing rules correspond as above to support functions on sets $Z_t$, which are general convex compact subsets of $\R^N_{++}$ (not necessarily polytopes).
\end{example}

\begin{example}
\label{ex: acceptance set}
Several concrete acceptance sets have been considered in the risk measure literature. In what follows we list some key examples when $\cX$ is a space of integrable random variables over a probability triple $(\Omega,\cF,\probp)$. For $\alpha\in(0,1)$ and $X\in\cX$ we denote by $q^+_\alpha(X)$ the upper $\alpha$-quantile of $X$. The next examples cover acceptance sets based on worst-case scenarios, Expected Shortfall, and expectiles (which are convex cones), acceptance sets based on expected utility and Adjusted Expected Shortfall (which are convex but not conic in general), and acceptance sets based on Value at Risk and Range Value at Risk (which are not convex in general). We refer to \cite{book:follmer2011stochastic} \cite{cont2010robustness}, \cite{article:bellini2014generalized}, \cite{burzoni2022adjusted} for more details about these acceptance sets.
\begin{itemize}
\item $\cA=\{X\in\cX \,; \ \probp(X\geq0)=1\}$.
\item $\cA=\big\{X\in\cX \,; \ \ES_\alpha(X):=-\frac{1}{\alpha}\int_0^\alpha q^+_\beta(X)d\beta\leq0\big\}$ for $\alpha\in(0,1)$.
\item $\cA=\left\{X\in\cX \,; \ \tfrac{\E_\probp[\max\{X,0\}]}{\E_\probp[\max\{-X,0\}]}\geq\frac{1-\alpha}{\alpha}\right\}$ where $\alpha\in(0,1/2)$ and we use the convention $\tfrac{0}{0}=\infty$.
\item $\cA= \{X\in \cX \,; \ \E_\probp[u(X)]\geq0\}$ where $u:\R\to[-\infty,\infty)$ is nondecreasing and satisfies $u(0)=0$.
\item $\cA=\{X\in\cX \,; \ \ES_\alpha(X)\leq g(\alpha), \ \forall \alpha\in(0,1)\}$ where $g:(0,1)\to(-\infty,\infty]$ is nonincreasing.
\item $\cA=\{X\in\cX \,; \ q^+_X(\alpha)\geq 0\}=\{X\in\cX \,; \ \probp(X<0)\leq\alpha\}$ where $\alpha\in(0,1)$.
\item $\cA=\big\{X\in\cX \,; \ \RVaR_{\alpha,\beta}(X):=-\frac{1}{\beta-\alpha}\int_\alpha^\beta q^+_\gamma(X)d\gamma\leq0\big\}$ for $0<\alpha<\beta<1$.
\end{itemize}
\end{example}


\section{Algebraic properties}
\label{sect: properties}

In this section we focus on a number of basic properties of risk measures encountered in the literature and show their link with the underlying financial primitives. This is important because it allows us to understand how a certain change in the market model impacts the behavior of a risk measure, thereby providing a stronger economic foundation for its properties. In addition, this helps recast the similar results in the literature as special cases of a general theory.

\smallskip

Our first result collects sufficient conditions for a risk measure to be star shaped, positively homogeneous, quasiconvex, convex, and subadditive. The terminology is standard and reviewed in the appendix. From a risk measure perspective, star shapedness and positive homogeneity mean that the ratio between risk capital and exposure size is increasing, respectively constant, while quasiconvexity, convexity, and subadditivity reflect, each in its own way, the diversification principle according to which risk capital of aggregate positions is controlled by risk capital of stand-alone positions. The combination of positive homogeneity and subadditivity was first studied in a risk measure setting in \cite{article:artzner1999coherent} and the extension to convex risk measures was taken up by \cite{article:follmer2002convex} and \cite{article:frittelli2002putting}. We refer to \cite{article:cerreia2011risk} and \cite{article:drapeau2013risk} for a treatment of quasiconvex risk measures and to \cite{castagnoli2021star} for a treatment of star-shaped risk measures.

\begin{proposition}
\label{prop: properties of rho}
The risk measure $\rho$ is nonincreasing. Moreover, the following statement hold:
\begin{enumerate}
  \item[(i)] If $\cA$ and $\cP$ are star shaped, $V_0$ is star shaped, and $V_1$ is anti-star shaped, then $\rho$ is star shaped.
\item[(ii)] If $\cA$ and $\cP$ are cones, and $V_0$ and $V_1$ are positively homogeneous, then $\rho$ is positively homogeneous.
  \item[(iii)] If $\cA$ and $\cP$ are convex, $V_0$ is quasiconvex, and $V_1$ is concave, then $\rho$ is quasiconvex.
  \item[(iv)] If $\cA$ and $\cP$ are convex, $V_0$ is convex, and $V_1$ is concave, then $\rho$ is convex.
  \item[(v)] If $\cA$ and $\cP$ are closed under addition, $V_0$ is subadditive, and $V_1$ is superadditive, then $\rho$ is subadditive.
\end{enumerate}
\end{proposition}
\begin{proof}
Take $X,Y\in\cX$ and assume that $X\geq Y$. If $x\in\cP$ is such that $Y+V_1(x)\in\cA$, then we get $X+V_1(x)\in\cA$ by monotonicity of $\cA$. As a result, $\rho(X)\leq\rho(Y)$, showing that $\rho$ is nonincreasing. Now, suppose the assumptions in (i) hold. Then, for $\lambda>1$ and $X\in\cX$
\begin{align*}
		\rho(\lambda X)&=\inf\{V_0(x) \,; \ x\in\cP, \ \lambda X+V_1(x)\in\cA\}\\
		&=\inf\left\{V_0(x) \,; \ x\in\cP, \ X+\tfrac{1}{\lambda}V_1(x)\in\cA\right\}\\
		&\geq \inf\left\{V_0(x) \,; \ x\in\cP, \ X+V_1\left(\tfrac{1}{\lambda}x\right)\in\cA\right\}\\
		&=\inf\{V_0(\lambda x) \,; \ x\in\cP, \ X+V_1(x)\in\cA\}\\
		&\geq \lambda\rho(X).
\end{align*}
This shows that $\rho$ is star shaped. Under the assumptions in (ii), the two inequalities are actually equalities, showing that $\rho$ is even positively homogeneous. Next, suppose the assumptions in (iii) hold. Take $X,Y\in\cX$ and $x,y\in\cP$ such that $X+V_1(x)\in\cA$ and $Y+V_1(y)\in\cA$. Moreover, take $\lambda\in[0,1]$. As
\begin{equation*}
		\lambda X+(1-\lambda)Y+V_1(\lambda x+(1-\lambda)y)\geq \lambda\big(X+V_1(x)\big)+(1-\lambda)\big(Y+V_1(y)\big)\in\cA,
\end{equation*}
we infer that $\rho(\lambda X+(1-\lambda)Y)\leq V_0(\lambda x+(1-\lambda)y)\leq \max\{V_0(x),V_0(y)\}$. This yields $\rho(\lambda X+(1-\lambda)Y)\leq\max\{\rho(X),\rho(Y)\}$ and shows that $\rho$ is quasiconvex. Under the assumptions in (iv), we additionally have $\rho(\lambda X+(1-\lambda)Y)\leq \lambda V_0(x)+(1-\lambda)V_0(y)$, showing that $\rho$ is convex. Finally, suppose the assumptions in (v) hold. Take $X,Y\in\cX$ and $x,y\in\cP$ such that $X+V_1(x)\in\cA$ and $Y+V_1(y)\in\cA$. Then
\begin{equation*}
X+Y+V_1(x+y)\geq \big(X+V_1(x)\big)+\big(Y+V_1(y)\big)\in\cA,
\end{equation*}
implying that $\rho(X+Y)\leq V_0(x+y)\leq V_0(x)+V_0(y)$. This shows that $\rho$ is subadditive.
\end{proof}

\begin{remark}
The conditions on the financial primitives listed in the preceding proposition are sufficient but, in general, not necessary to yield the corresponding properties of risk measures. For instance, we show a simple example where $\rho$ is convex even though $\cA$ is not convex. Let $\cX$ be the space of random variables on the probability space $(\Omega,\cF,\probp)$ where $\Omega=\{\omega_1,\omega_2\}$, $\cF$ coincides with the parts of $\Omega$, and $\probp$ assigns probability $1/2$ to each scenario. Set $\cP=\R^2$ and define $\cA=\cA_1\cup\cA_2\cup\cA_3$ where
\[
\cA_1 = \{X\in\cX \,; \ X(\omega_1)\geq0, \ X(\omega_2)\geq0\},
\vspace{-0.12cm}
\]
\[
\cA_2 = \{X\in\cX \,; \ -1\leq X(\omega_1)\leq0, \ X(\omega_2)\geq1\},
\]
\[
\cA_3 = \{X\in\cX \,; \ X(\omega_1)\geq1, \ -1\leq X(\omega_2)\leq0\}.
\]
Observe that $\cA$ is not convex. In fact, it is not even star shaped. Moreover, define
\[
V_0(x)=x_1+x_2, \ \ \ V_1(x)=x_1R+x_2S, \ \ \ x\in\R^2,
\]
where $R=1_\Omega\in\cX$ and $S=1_{\{\omega_1\}}\in\cX$. For $i=1,2,3$ and for every $X\in\cX$
\[
\rho_i(X)=\inf\{V_0(x) \,; \ x\in\R^2, \ X+V_1(x)\in\cA_i\}=
\begin{cases}
-X(\omega_1) & \mbox{if} \ i=1,\\
-X(\omega_1)-1 & \mbox{if} \ i=2,\\
-X(\omega_1)+1 & \mbox{if} \ i=3.
\end{cases}
\]
As a result, $\rho(X)=\min\{\rho_1(X),\rho_2(X),\rho_3(X)\}=-X(\omega_1)-1$. This shows that $\rho$ is convex.
\end{remark}

We establish a representation of general risk measures in terms of convex risk measures that is akin to a classical representation of Value at Risk, see Proposition 4.47 in~\cite{book:follmer2011stochastic}, which was recently generalized in Theorem 5 in~\cite{castagnoli2021star}. More precisely, we show that, if the underlying market frictions are convex, then any risk measure can be expressed as a minimum of convex risk measures with the same market primitives. As discussed in the literature, this type of convex representation is useful in optimization problems, e.g., risk minimization. Note that our result is stated for general, not necessarily star shaped, acceptance sets.

\begin{proposition}
If $\cP$ is convex, $V_0$ is convex, and $V_1$ is concave, then there exists a set $I$ such that for every $i\in I$ we find a convex acceptance set $\cA_i\subset\cX$ satisfying
\[
\rho(X) = \min_{i\in I}\rho_i(X), \ \ \ X\in\cX,
\]
where $\rho_i$ is the convex risk measure associated with $(\cA_i,\cP,V_0,V_1)$, i.e.
\[
\rho_i(X) = \inf\{V_0(x) \,; \ x\in\cP, \ X+V_1(x)\in\cA_i\}, \ \ \ X\in\cX.
\]
\end{proposition}
\begin{proof}
For every $Y\in\cA$ consider the convex acceptance set $\cA_Y=\{X\in\cX \,; \ X\geq Y\}$. Moreover, define
\[
\rho_Y(X) = \inf\{V_0(x) \,; \ x\in\cP, \ X+V_1(x)\in\cA_Y\}, \ \ \ X\in\cX.
\]
It follows from Proposition~\ref{prop: properties of rho} that $\rho_Y$ is convex. Note that $\cA=\bigcup_{Y\in\cA}\cA_Y$. As a result, for every $X\in\cX$
\[
\rho(X) \leq \inf_{Y\in\cA}\rho_Y(X).
\]
To conclude the proof, take $X\in\cX$ and let $x\in\cP$ satisfy $X+V_1(x)\in\cA$. Then, there exists $Z\in\cA$ such that $X+V_1(x)\in\cA_Z$, which implies $\rho_Z(X)\leq V_0(x)$. Taking the infimum over $x$ delivers $\rho_Z(X) \leq \rho(X)$ and proves the desired assertion.
\end{proof}

The next proposition deals with the situation where some of the eligible assets is traded in a frictionless way. In this case, the corresponding risk measure enjoys a translation invariance property that extends the standard property of cash additivity and its generalizations in a frictionless market, e.g., the property of $S$-additivity studied in \cite{article:farkas2014beyond}. Note that the pricing rules defined in Example~\ref{ex: kabanov 1p} are additive in the sense below.

\begin{proposition}
Let $z\in\cP$ satisfy $\cP+\Span(\{z\})\subset\cP$ and let $V_0$ and $V_1$ be additive along $z$, i.e.
\[
V_0(x+\lambda z)=V_0(x)+\lambda V_0(z), \ \ \ V_1(x+\lambda z)=V_1(x)+\lambda V_1(z), \ \ \ x\in\cP, \ \lambda\in\R.
\]
Then, $\rho$ is price additive along $z$, i.e.
\begin{equation*}
    \rho(X+\lambda V_1(z)) = \rho(X)-\lambda V_0(z), \ \ \ X\in\cX, \ \lambda\in\R.
\end{equation*}
\end{proposition}
\begin{proof}
Let $X\in\cX$ and $\lambda\in\R$. By assumption, we can write
\begin{align*}
\rho(X+\lambda V_1(z)) &=
\inf\{V_0(x) \,; \ x\in\cP, \ X+\lambda V_1(z)+V_1(x)\in\cA\}\\
&= \inf\{V_0(x) \,; \ x\in\cP, \ X+V_1(\lambda z+x)\in\cA\}\\
&= \inf\{V_0(x-\lambda z) \,; \ x\in\cP, \ X+V_1(x)\in\cA\}\\
&= \inf\{V_0(x)-\lambda V_0(z) \,; \ x\in\cP, \ X+V_1(x)\in\cA\}\\
&= \rho(X)-\lambda V_0(z).\qedhere
\end{align*}
\end{proof}

The previous result can be used to derive two special representations of risk measures. Both of them are expressed in terms of the set of liquidation values of zero-cost admissible portfolios, i.e.
\[
\cV := \{V_1(x) \,; \ x\in\cP, \ V_0(x)=0\}.
\]
The first representation shows that, if some of the eligible assets is frictionless, the corresponding risk measure determines the minimal amount of capital to raise and invest in that asset in order to ensure acceptability with respect to the enlarged acceptance set $\cA-\cV$. This set consists of all financial positions that can be made acceptable up to liquidation of a zero-cost admissible portfolio. The result shows that it is therefore enough to find a frictionless eligible asset to be able to formally reduce a general risk measure to a risk measure with respect to a single eligible asset as studied in \cite{article:farkas2014beyond}. In particular, the result extends the reduction lemma in \cite{article:farkas2015measuring} obtained in a frictionless multi-asset setting and is related to the class of risk measures studied in \cite{cheridito2017duality}.

\begin{corollary}
Let $z\in\cP$ satisfy $\cP+\Span(\{z\})\subset\cP$ and let $V_0$ and $V_1$ be additive along $z$. If $V_0(z)>0$,
\[
\rho(X)=\inf\{\lambda V_0(z) \,; \ \lambda\in\R, \ X+\lambda V_1(z)\in\cA-\cV\}, \ \ \ X\in\cX.
\]
\end{corollary}
\begin{proof}
Let $X\in\cX$. Note that for every $x\in\cP$ there exist $y\in\cP$ and $\lambda\in\R$ such that $V_0(y)=0$ and $x=y+\lambda z$. Indeed, by our assumptions, it suffices to take $\lambda=V_0(x)/V_0(z)$ and $y=x-\lambda z$. As a result,
\begin{align*}
\rho(X)
&=
\inf\{V_0(y+\lambda z) \,; \ y\in\cP, \ V_0(y)=0, \ \lambda\in\R, \ X+V_1(y+\lambda z)\in\cA\} \\
&=
\inf\{\lambda V_0(z) \,; \ \lambda\in\R, \ Y\in\cV, \ X+Y+\lambda V_1(z)\in\cA\} \\
&=
\inf\{\lambda V_0(z) \,; \ \lambda\in\R, \ X+\lambda V_1(z)\in\cA-\cV\}.\qedhere
\end{align*}
\end{proof}

The second representation holds in spaces of random variables under the assumption that some of the eligible assets is risk free. In this case, the risk measure can be expressed as a suitable worst-case market version of the standard cash-additive risk measure $\rho_\cA$ induced by the underlying acceptance set.

\begin{corollary}
Let $\cX$ be a space of random variables over the probability space $(\Omega,\cF,\probp)$ and define
\[
\rho_\cA(X) := \inf\{m\in\R \,; \ X+m1_\Omega\in\cA\}, \ \ \ X\in\cX.
\]
Suppose that $\{X\in\cX \,; \ \rho_\cA(X)\leq0\}\subset\cA$. Let $z\in\cP$ satisfy $\cP+\Span(\{z\})\subset\cP$ and let $V_0$ and $V_1$ be additive along $z$. If $V_0(z)>0$ and $V_1(z)=1_\Omega$, then
\[
\rho(X) = V_0(z)\sup_{Y\in\cV}\rho_\cA(X+Y), \ \ \ X\in\cX.
\]
\end{corollary}
\begin{proof}
Let $X\in\cX$. It follows from the preceding corollary that
\begin{align*}
\rho(X)
&=
V_0(z)\inf\{\lambda\in\R \,; \ X+\lambda1_\Omega\in\cA-\cV\} \\
&=
V_0(z)\inf\{\lambda\in\R \,; \exists Y\in\cV, \ X+Y+\lambda1_\Omega\in\cA\} \\
&=
V_0(z)\inf\{\lambda\in\R \,; \exists Y\in\cV, \ \rho_\cA(X+Y+\lambda1_\Omega)\leq0\} \\
&=
V_0(z)\inf\{\lambda\in\R \,; \exists Y\in\cV, \ \rho_\cA(X+Y)\leq\lambda\} \\
&=
V_0(z)\sup_{Y\in\cV}\rho_\cA(X+Y).\qedhere
\end{align*}
\end{proof}



\section{Continuity properties}
\label{sect: lsc}

In this section we study a key regularity property of risk measures, namely lower semicontinuity, which, as discussed in the next section, is crucial to obtain dual representations of convex and quasiconvex risk measures. The goal is to provide a variety of sufficient conditions for lower semicontinuity. This is considerably more challenging than the study of algebraic properties and requires some preparation. In a first step, we introduce suitable extensions of the classical ``no arbitrage'' condition. In a second more technical step, we have to focus on the regularity of the liquidation pricing rule. Finally, we will be able to establish the desired sufficient conditions for lower semicontinuity.


\subsection{Acceptable deals}
\label{sect: acceptable deals}

An acceptable deal is any admissible portfolio of basic securities that can be purchased at zero cost at the initial date and delivers a nonzero acceptable liquidation value at the terminal date. As such, an acceptable deal is a natural extension of an arbitrage opportunity, where the reference positive cone is replaced by the set of acceptable positions, and is closely related to the notion of good deal studied, e.g., in \cite{article:cochrane2000beyond}, \cite{article:bernardo2000gain}, \cite{article:carr2001pricing}, \cite{article:jaschke2001coherent}. We also introduce the notion of a scalable acceptable deal, which can be viewed as an asymptotic version of an acceptable deal. Our terminology is inspired by \cite{article:pennanen2011arbitrage}, where the focus is on markets with convex frictions and the acceptance set is the standard positive cone. We also refer to \cite{article:arduca2020market}, where the focus is on markets with convex frictions and convex acceptance sets and the eligible assets pay off at the terminal date so that $V_1$ is be taken to be linear.

\begin{definition}
We say that a portfolio $x\in\R^N$ is:
\begin{enumerate}
  \item[(1)] an {\em acceptable deal} if $x\in\cP$, $V_0(x)\leq0$, and $V_1(x)\in\cA\setminus\{0\}$.
  \item[(2)] a {\em scalable acceptable deal} if $x\in\cP^\infty$, $V^\infty_0(x)\leq0$, and $V_1(x)\in\cA^\infty\setminus\{0\}$.
\end{enumerate}
We replace the term ``(scalable) acceptable deal'' with ``(scalable) arbitrage opportunity'' when $\cA=\cX_+$.
\end{definition}

As illustrated by our example below, a scalable acceptable deal need not be an acceptable deal in general. However, the next proposition shows that, in many standard situations, every scalable acceptable deal is an acceptable deal, and the set of scalable acceptable deals consists of those portfolios that are admissible and available at zero cost independently of their size and whose liquidation value is acceptable independently of its size. This justifies the chosen terminology.

\begin{proposition}
Assume that $\cA$ is closed and star shaped, $\cP$ is closed and star shaped, and $V_0$ is lower semicontinuous and star shaped. Then, for every scalable acceptable deal $x\in\R^N$ and every $\lambda>0$ we have $\lambda x\in\cP$, $V_0(\lambda x)\leq0$, and $\lambda V_1(x)\in\cA$. In particular, $x$ is an acceptable deal.
\end{proposition}
\begin{proof}
It is not difficult to show that $\cP^\infty\subset\cP$ and $\cA^\infty\subset\cA$ by closedness and star shapedness. Similarly, $V_0\leq V_0^\infty$ by lower semicontinuity and star shapedness. Now, let $x\in\R^N$ be a scalable acceptable deal and $\lambda>0$. As $x\in\cP^\infty\subset\cP$ and $\cP^\infty$ is a cone, we have $\lambda x\in\cP$. Similarly, as $V_1(x)\in\cA^\infty\subset\cA$ and $\cA^\infty$ is a cone, we have $\lambda V_1(x)\in\cA$. To conclude the proof, observe that $V_0(\lambda x)\leq V^\infty_0(\lambda x)=\lambda V^\infty_0(x)\leq0$ by positive homogeneity of $V^\infty_0$.
\end{proof}

\begin{example}
Let $\cX$ be the space of random variables on the probability space $(\Omega,\cF,\probp)$, where $\Omega=\{\omega_1,\omega_2\}$, $\cF$ coincides with the parts of $\Omega$, and $\probp$ assigns probability $1/2$ to each scenario. Set $\cP=\R^2$ and
\[
\cA=\{X\in\cX \,; \ X(\omega_1)\geq0, \ X(\omega_2)\geq0\}\cup\bigcup_{n\in\N}\{X\in\cX \,; \ -n-1\leq X(\omega_1)\leq-n, \ X(\omega_2)\geq n\},
\]
\[
V_0(x)=x_1+x_2, \ \ \ V_1(x)=(x_1,x_2), \ \ \ x\in\R^2.
\]
It can be verified that $\cA^\infty=\{X\in\cX \,; \ X(\omega_1)+X(\omega_2)\geq0, \ X(\omega_2)\geq0\}$. Take $x=(-1/2,1/2)\in\R^2$. We immediately see that $x$ is a scalable acceptable deal. However, $V_1(x)\notin\cA$, showing that it is not an acceptable deal. The problem is that $\cA^\infty$ is not included in $\cA$ in this case.
\end{example}

We highlight necessary and sufficient conditions for the absence of scalable acceptable deals in terms of
\[
\cL := \{x\in\cP^\infty \,; \ V^\infty_0(x)\leq0, \ V_1(x)\in\cA^\infty\}.
\]
This set ``almost coincides'' with the set of acceptable deals and will play an important role in the sequel. Here, we set $\ker(V_1):=\{x\in\R^N \,; \ V_1(x)=0\}$.

\begin{proposition}
\label{prop: conditions on L}
Consider the following statements:
\begin{enumerate}
  \item[(i)] There exists no scalable acceptable deal.
  \item[(ii)] $\cL$ is a linear space.
  \item[(iii)] $\cL=\{0\}$.
\end{enumerate}
Then, (iii)$\implies$(i). If $\ker(V_1)=\{0\}$, then (i)$\implies$(iii). If $\cA^\infty\cap(-\cA^\infty)=\{0\}$, then (ii)$\implies$(i).
\end{proposition}
\begin{proof}
It is clear that (iii) implies (i). Now, suppose that $\ker(V_1)=\{0\}$ and there exists no scalable acceptable deal. In this case, we must have $\cL\subset\ker(V_1)$ and we conclude that $\cL=\{0\}$, showing that (i) implies (iii). Next, assume that $\cL$ is a linear space and take $x\in\cL$. Observe that $-x\in\cL$. Moreover, note that $\cA^\infty+\cX_+\subset\cA^\infty$. By our standing assumptions, we infer that $-V_1(x)\geq V_1(-x)\in\cA^\infty$ and, hence, $-V_1(x)\in\cA^\infty$. This implies that $V_1(x)\in\cA^\infty\cap(-\cA^\infty)=0$, showing that no scalable acceptable deal can exist. As a result, (ii) implies (i).
\end{proof}

\begin{remark}
(i) The condition $\ker(V_1)=\{0\}$ is satisfied in a variety of situations of interest. For instance, in the setting of Example~\ref{ex: pricing rules}, the condition holds if $\varphi_1,\dots,\varphi_N$ are strictly increasing and $S^a_1,\dots,S^a_N,S^b_1,\dots,S^b_N$ are linearly independent.

\smallskip

(ii) The pointedness condition $\cA^\infty\cap(-\cA^\infty)=\{0\}$ is satisfied by many relevant examples of acceptance sets. We refer to Proposition 5.9 in \cite{article:bellini2020law} for general sufficient conditions for pointedness in spaces of random variables.
\end{remark}

The preceding proposition is sharp. Indeed, we show that the linearity of $\cL$ does not imply absence of scalable acceptable deals. In addition, the absence of scalable acceptable deals does not generally imply that $\cL$ is reduced to zero or that $\cL$ is a linear space. This is true even if $\cA^\infty$ satisfies the pointedness condition in the proposition.

\begin{example}
Let $\cX$ be the space of random variables on the probability space $(\Omega,\cF,\probp)$, where $\Omega=\{\omega_1,\omega_2\}$, $\cF$ coincides with the parts of $\Omega$, and $\probp$ assigns probability $1/2$ to each scenario.

\smallskip

(i) Set $\cA=\{X\in\cX \,; \ X(\omega_1)\geq0\}$ and define $\cP=\R^2$ and
\[
V_0(x)=x_1+x_2, \ \ \ V_1(x)=(x_1+x_2,x_2), \ \ \ x\in\R^2.
\]
It is easy to verify that $\cL=\{x\in\R^2 \,; \ x_1+x_2=0\}$, showing that $\cL$ is a linear space. However, $x=(-1,1)\in\R^2$ is a scalable acceptable deal.

\smallskip

(ii) Set $\cA=\{X\in\cX \,; \ X(\omega_1)\geq0, \ X(\omega_2)\geq0\}$ and define $\cP=\R^2$ and
\[
V_0(x)=x_1+2x_2, \ \ \ V_1(x)=(\min\{x_1+2x_2,2x_1+x_2\},x_1+2x_2), \ \ \ x\in\R^2.
\]
It is immediate to verify that no scalable acceptable deal exists. However, the portfolio $x=(2,-1)\in\R^2$ belongs to $\cL$ but not to $-\cL$, showing that $\cL$ is not linear.
\end{example}

The conditions appearing in the preceding proposition are going to be crucial for our study of dual representations of risk measures. We conclude this section by discussing their economic rationale. As already mentioned, the absence of acceptable deals can be seen as a generalization of the absence of arbitrage opportunities. However, there is a fundamental difference between the two. While every agent recognizes an arbitrage opportunity as such and will try to exploit it, there might be no consensus across agents in the identification of a common criterion of acceptability. Hence, at first sight, postulating the absence of acceptable deals does not seem to be economically motivated and, in the best case, calls for a deeper analysis. The good news is that, in order to develop our key results, we are not forced to work under the absence of acceptable deals but it will be sufficient to stipulate assumptions about the set $\cL$, which are related to the weaker absence of scalable acceptable deals. As shown by the next proposition, this condition is satisfied in a number of relevant situations and is sometimes automatically implied by the absence of (scalable) arbitrage opportunities, e.g., when short selling is either prohibited or restricted for every basic security.

\begin{proposition}
\label{prop: sufficient conditions for NSAD}
Assume that any of the following conditions holds:
\begin{enumerate}
  \item[(i)] $\cA^\infty\subset\cX_+$ and there exists no scalable arbitrage opportunity.
  \item[(ii)] $\cP^\infty\subset\R^N_+$, $V_1(x)\in\cX_+$ for every $x\in\R^N_+$, and there exists no scalable arbitrage opportunity.
  \item[(iii)] $\cP$ is bounded.
\end{enumerate}
Then, there exists no scalable acceptable deal. Under (iii), we even have $\cL=\{0\}$.
\end{proposition}
\begin{proof}
Take $x\in\cL$ and recall that $x$ belongs to $\cP^\infty$ and satisfies $V^\infty_0(x)\leq0$ and $V_1(x)\in\cA^\infty$. Under (i) or (ii), we have $V_1(x)\in\cX_+$. As there exists no scalable arbitrage opportunity, we must have $V_1(x)=0$. Hence, there exists no scalable acceptable deal either. Under (iii), we easily see that $\cP^\infty=\{0\}$ holds, showing that $\cL=\{0\}$ and implying absence of scalable acceptable deals.
\end{proof}


\subsection{Regularity of the liquidation rule}
\label{subsect: usc}

To obtain sufficient conditions for lower semicontinuity of our general risk measures, the liquidation pricing rule $V_1$ has to display some degree of regularity. The crucial property turns out to be upper semicontinuity.

\begin{definition}
We say that $V_1$ is upper semicontinuous at $x\in\R^N$ if for every neighborhood $\cU$ of $V_1(x)$ there exists a neighborhood $\cV$ of $x$ such that $V_1(\cV)\subset \cU-\cX_+$. We say that $V_1$ is upper semicontinuous if it is upper semicontinuous at every $x\in\R^N$. 
\end{definition}

It is clear that our notion of upper semicontinuity extends the usual notion for real-valued maps. The next proposition records some characterizations of upper semicontinuity that are repeatedly used later on without explicit reference.

\begin{proposition}
\label{prop: characterizations usc V1}
Let $x\in\R^N$. The following statements are equivalent:
\begin{enumerate}
        \item[(i)] $V_1$ is upper semicontinuous at $x$.
        \item[(ii)] For every net $(x_\alpha)\subset\R^N$ such that $x_\alpha\to x$ and for every neighborhood $\cU$ of $V_1(x)$ there exists $\alpha_\cU$ such that, if $\alpha\succeq \alpha_\cU$, then $V_1(x_\alpha)\leq Y_\alpha$ for some $Y_\alpha\in\cU$.
        \item[(iii)] For every net $(x_\alpha)\subset\R^N$ such that $x_\alpha\to x$ there exists a subnet $(x_{\beta})\subset(x_\alpha)$ and $(Y_\beta)\subset\cX$ such that $Y_\beta\to V_1(x)$ and $V_1(x_\beta)\leq Y_\beta$ for every $\beta$.
\end{enumerate}
If $\cX$ is first countable, the previous statements are also equivalent to:
\begin{enumerate}
        \item[(iv)] For every sequence $(x_n)\subset\R^N$ such that $x_n\to x$ there exists $(Y_n)\subset\cX$ such that $Y_n\to V_1(x)$ and $V_1(x_n)\leq Y_n$ for every $n\in\N$.
\end{enumerate}
If $V_1(\R^N)$ is contained in a finite-dimensional space with dimension $m$, then $V_1=(V_{1,1},\dots,V_{1,m})$ for suitable functions $V_{1,1},\dots,V_{1,m}:\R^N\to\R$. In this case, the previous statements are also equivalent to:
\begin{enumerate}
        \item[(v)] The functions $V_{1,1},\dots,V_{1,m}$ are upper semicontinuous at $x$.
\end{enumerate}
\end{proposition}
\begin{proof}
It is clear that (i) implies (ii). To show that (ii) implies (iii), take a net $(x_\alpha)\subset\R^N$ such that $x_\alpha\to x$. Let $A$ be the corresponding index set. The set $B$ of couples $(\alpha,\cU)$ where $\cU$ is a neighborhood of $V_1(x)$ and $\alpha\in A$ satisfies $\alpha\succeq\alpha_\cU$ is directed by the binary relation $(\alpha,\cU)\geq(\alpha',\cU')$ if and only if $\alpha\succeq\alpha'$ and $\cU\subseteq\cU'$. It is easily verified that the net with generic term $x_{(\alpha,\cU)}=x_\alpha$ is a subnet of $(x_\alpha)$ with index set $B$. By (ii), for every $(\alpha,\cU)\in B$, we find $Y_{(\alpha,\cU)}\in\cU$ such that $V_1(x_{(\alpha,\cU)})\leq Y_{(\alpha,\cU)}$ and clearly $Y_{(\alpha,\cU)}\to V_1(x)$. This concludes the proof of the implication. It remains to prove that (iii) implies (i). By contradiction, assume that $V_1$ is not upper semicontinuous at $x$. Let $\cU$ be a neighborhood of $V_1(x)$ such that every neighborhood $\cV$ of $x$ contains $x_\cV$ such that $V_1(x_\cV)$ is not dominated by any point of $\cU$. Let $(\cU_n)$ be a fundamental system of neighborhoods of $x$. By (iii), there is $Y\in\cU$ such that $V_1(x_{\cU_k})\leq Y$ for some $k\in\N$, which is a contradiction. Now, assume that $\cX$ is first countable. Clearly, (iv) implies (iii). We prove that (ii) implies (iv). Take $(x_n)\subset\R^N$ with $x_n\to x$ and let $(\cU_k)$ be a fundamental system of neighborhood of $V_1(x)$ such that $\cU_k\supset\cU_{k+1}$ for every $k\in\N$. From (ii), we know that, for every $k\in\N$, we find a corresponding $n_k\in\N$ such that for $n\geq n_k$ there is $Y^n_k\in\cU_k$ with $V_1(x_n)\leq Y^n_k$. We can assume that $(n_k)$ is strictly increasing. The sequence we are seeking is defined by $Y_n=Y^n_k$ if $n_k\leq n<n_{k+1}$. This completes the proof. The equivalence between (iv) and (v) under finite dimensionality of the range of $V_1$ is straightforward.
\end{proof}

We exhibit a number of situations where $V_1$ satisfies upper semicontinuity. The first result follows directly from point (v) of Proposition \ref{prop: characterizations usc V1} and shows that the standard separable liquidation rules considered in frictionless markets or in markets with frictions, see Example~\ref{ex: pricing rules}, are upper semicontinuous provided their ``components'' are standard upper semicontinuous functionals.

\begin{proposition}
\label{prop: V1 usc, trans costs}
Let $X_1,\dots,X_m\in\cX_+$ and $\varphi_1,\dots,\varphi_m:\R^N\to\R$ satisfy
\[
V_1(x)=\sum_{i=1}^m \varphi_i(x)X_i, \ \ \ x\in\R^N.
\]
If $\varphi_1,\dots,\varphi_m$ are upper semicontinuous, then $V_1$ is upper semicontinuous.
\end{proposition}

The next sufficient condition for upper semicontinuity holds also for liquidations rules that do not possess an additive structure. For instance, it applies to the liquidation pricing rules described in Example \ref{ex: kabanov 1p}. We say that $\cX$ is locally solid if there exists a fundamental system of neighborhoods of $0$ consisting of solid sets, i.e., sets $\cS\subset\cX$ such that $X\in\cS$ whenever $\vert X\vert\leq\vert Y\vert$ for some $Y\in\cS$. Every Banach lattice is locally solid with respect to its norm topology. In particular, every Orlicz space (equipped with the natural almost-sure partial order) is locally solid with respect to the Luxemburg norm topology.

\begin{proposition}
\label{prop: automatic continuity}
Let $\cX$ be locally solid. If $V_1$ is concave and monotone increasing, then it is upper semicontinuous.
\end{proposition}
\begin{proof}
As $\cX$ is locally solid, there exists a fundamental system of neighborhoods $\cU$ of $0$ such that $X\in\cU$ whenever $0\leq X\leq Y$ for some $Y\in\cU$. It follows from Proposition 1.3 in \cite{peressini1967ordered} that $\cX_+$ is a normal cone. As a direct consequence of the automatic continuity result in Corollary 2.4 in \cite{article:borwein1987automatic}, we conclude that $V_1$ is continuous and, hence, upper semicontinuous.
\end{proof}

We conclude with additional sufficient conditions for upper semicontinuity that are specific to the case where $\cX$ is contained in $L^0$, the space of (equivalence classes under almost-sure equality of) random variables on a given probability space equipped with the standard almost-sure partial order. In this case, we say that $V_1$ is weakly upper semicontinuous if for all $(x_n)\subset\R^N$ and $x\in\R^N$ such that $x_n\to x$ we have $V_1(x)\geq \limsup_{n\to\infty}V_1(x_n)$ almost surely. We also say that $\cX$ in an ideal of $L^0$ if it is solid and $\max\{X,Y\}\in\cX$ for all $X,Y\in\cX$. In addition, $\cX$ is order continuous if for all $(X_n)\subset\cX$ and $X\in\cX$
\[
X_n\to X \ \mbox{almost surely}, \ \sup_{n\in\N}\vert X_n\vert\in\cX \ \implies \ X_n\to X.
\]
By dominate convergence, every Orlicz heart is order continuous with respect to the Luxemburg norm. In particular, every $L^p$ space for $1\leq p<\infty$ is order continuous with respect to its standard $p$ norm. The space $L^\infty$ is order continuous when equipped with its weak-star topology.

\begin{proposition}
Let $(\Omega,\cF,\probp)$ be a probability space and let $\cX$ be an ideal of $L^0$. Assume that $V_1$ is weakly upper semicontinuous and that one of the following conditions holds:
\begin{enumerate}
    \item[(i)] $\cX$ coincides with $L^0$ and is equipped with the topology of convergence in probability.
    \item[(ii)] $\cX$ is order continuous and $V_1$ is nondecreasing.
\end{enumerate}
Then, $V_1$ is upper semicontinuous.
\end{proposition}
\begin{proof}
Take $(x_n)\subset\R^N$ and $x\in\R^N$ such that $x_n\to x$ and define $Y_n=\max\{V_1(x_n),V_1(x)\}\in\cX$ for every $n\in\N$. It follows from weak upper semicontinuity that $V_1(x)\geq \limsup_{n\to\infty}V_1(x_n)$ almost surely. This implies that $Y_n\to V_1(x)$ almost surely. To show that $V_1$ is upper semicontinuous at $x$ under (i), it suffices to observe that $Y_n\to V_1(x)$ with respect to convergence in probability as well and to apply Proposition \ref{prop: characterizations usc V1}. Now, assume that (ii) holds. Note that $\vert Y_n\vert\leq\max\{\vert V_1(x_n)\vert,\vert V_1(x)\vert\}$ almost surely for every $n\in\N$. As $(x_n)$ is bounded, it follows from the monotonicity of $V_1$ that $\vert V_1(x_n)\vert\leq Z$ almost surely for a suitable $Z\in\cX$. This implies that $\sup_{n\in\N}\vert Y_n\vert\in\cX$. By order continuity, $Y_n\to V_1(x)$, proving that $V_1$ is upper semicontinuous at $x$ due to Proposition \ref{prop: characterizations usc V1}.
\end{proof}

\begin{remark}
We show that $V_1$ needs not be upper semicontinuous if $L^0$ is replaced in point (i) above by a smaller space. To see this, let $\cX=L^1([0,1])$ be the space of equivalence classes (modulo almost-sure equality under the Lebesgue measure) of integrable functions on $[0,1]$. We equip $\cX$ with its canonical norm structure. Let $N=1$ and $\cP=\R$ and define for all $x\in\R$ and $\omega\in[0,1]$
\[
V_1(x)(\omega)=
\begin{cases}
x\omega^{\vert x\vert-1} & \text{if} \ \omega>0,\\
0 & \text{if} \ \omega=0.
\end{cases}
\]
Clearly, $V_1$ is weakly upper semicontinuous. However, $V_1$ is not upper semicontinuous at $0$. To see this, note that for every $x\in(0,1]$
\begin{equation}
\label{eq: norm one}
\norm{V_1(x)}_1 = \int_0^1 x\omega^{x-1}d\omega = 1.
\end{equation}
Now, let $x_n=\frac{1}{n}\in\cP$ for $n\in\N$. We have $x_n\to 0$ but for every $(Y_n)\subset\cX$ satisfying $Y_n\geq V_1(x_n)\geq0$ almost surely for every $n\in\N$, one cannot have $Y_n\to V_1(0)$, for otherwise $V_1(x_n)\to V_1(0)=0$ as well, contradicting~\eqref{eq: norm one}.
\end{remark}


\subsection{Sufficient conditions for lower semicontinuity}

This section is devoted to establish sufficient conditions for our general risk measures to be lower semicontinuous. Our strategy builds on the following equivalent way to express a risk measure. The simple proof is omitted.
 
\begin{lemma}
\label{lem: defn C}
Let $\cC := \{(X,m)\in\cX\times\R \,; \ \exists x\in\cP \,:\, V_0(x)\leq m, \ X+V_1(x)\in\cA\}$. Then,
\begin{equation}
    \label{eq:rho inf C}
    \rho(X) = \inf\{m\in\R \,; \ (X,m)\in\cC \}, \ \ \ X\in\cX.
\end{equation}
\end{lemma}

\begin{remark}
\label{rem: epigraph}
The set $\cC$ is closely related to the epigraph $\epi(\rho)$ of $\rho$. Indeed, $\cC\subset\epi(\rho)\subset\cl(\cC)$, where $\cl(\cC)$ denotes the closure of $\cC$ in the natural product topology. The second inclusion follows directly from
\begin{equation*}
\{X\in\cX \,; \ \rho(X)\leq m\} = \bigcap_{k>m}\{X\in\cX \,; \ (X,k)\in\cC\} \ \ \  \text{for every }m\in\R.
\end{equation*}
Note that the inclusions $\cC\subset\epi(\rho)\subset\cl(\cC)$ are strict in general. To see this, let $\cX$ be the space of random variables on the probability space $(\Omega,\cF,\probp)$, where $\Omega=\{\omega_1,\omega_2\}$, $\cF$ coincides with the parts of $\Omega$, and $\probp$ assigns probability $1/2$ to each scenario. Set $\cP=\R$ and define $\cA=\{X\in\cX \,; \ X(\omega_1)>-1, \ X(\omega_2)>0\}$. Moreover, let $S=1_{\{\omega_1\}}\in\cX$ and define
\[
V_0(x)=x, \ \ \ V_1(x)=xS, \ \ \ x\in\R.
\]
It is easy to verify that for every $X\in\cX$
\[
\rho(X)=
\begin{cases}
-X(\omega_1)-1 & \mbox{if} \ X(\omega_2)>0,\\
\infty & \mbox{if} \ X(\omega_2)\leq0.
\end{cases}
\]
To conclude, set $U=1_\Omega\in\cX$ and observe that $(U,-2)\in\epi(\rho)\setminus\cC$ and $(-S,0)\in\cl(\cC)\setminus\epi(\rho)$. In particular, setting $X_n=\frac{1}{n}U-S\in\cX$ for $n\in\N$, we have that $(X_n,0)\in\cC$ for every $n\in\N$ and $(X_n,0)\to(-S,0)$.
\end{remark}

The set $\cC$ consists of all couples $(X,m)$ such that the budget $m$ is sufficient to make the financial position $X$ acceptable upon acquisition of an admissible portfolio of eligible assets. Its link with lower semicontinuity is clear: It follows from the previous lemma that $\rho$ is lower semicontinuous whenever $\cC$ is closed. The next theorem is the main result of this section and records sufficient conditions for $\cC$ to be closed. Besides some mild and widely satisfied requirements on the primitives $(\cA,\cP,V_0,V_1)$, we need to assume suitable ``no arbitrage'' conditions, which were discussed in Section~\ref{sect: acceptable deals}, together with upper semicontinuity of $V_1$, for which we refer to Section~\ref{subsect: usc}. It is worthwhile noting that convexity is not necessary to establish closedness and, hence, lower semicontinuity. Before stating the theorem, we single out the following projection lemma which is used in the proof. Here, we define
\[
\cN := \cL\cap(-\cL),
\]
and denote by $\cN^\perp$ the orthogonal complement of $\cN$ in $\Span(\cP)$. The set $\cL$ is linked to scalable acceptable deals and was defined in Section~\ref{sect: acceptable deals}.

\begin{lemma}
\label{lemma:wlog_lambda_N0rt}
If $\cA$ is convex and closed, $\cP$ is convex and closed, $V_0$ is convex and lower semicontinuous, and $V_1$ is superlinear, then $\cN$ is a linear space and for every $(X,m)\in\cC$ there exists $x\in {\cN}^\perp$ such that $V_0(x)\leq m$ and $X+V_1(x)\in\cA$.
\end{lemma}
\begin{proof}
We use throughout the properties of asymptotic cones and functions recalled in the appendix. It follows from our assumptions that $\cL$ is a cone, so that $\cN$ is a linear space. Now, take $(X,m)\in\cC$ and note that we find $y\in\cP$ such that $V_0(y)\leq m$ and $X+V_1(y)\in\cA$. Let $z\in\cN$ be the orthogonal projection of $y$ on $\cN$ and define $x=y-z\in {\cN}^\perp$. As $-\cN\subset\cP^\infty$ and $\cP$ is convex, closed, and contains $0$, we have $x\in\cP+\cP^\infty\subset\cP$. Moreover, by convexity and lower semicontinuity of $V_0$, it holds that $V_0(x)=V_0(y-z)\leq V_0(y)\leq m$. Here, we used that $V^\infty_0(-z)\leq0$, so that $V_0(y-z)-V_0(y)\leq0$. Finally, using superlinearity of $V_1$ and the fact that $\cA$ is convex, closed, and contains $0$, we infer that $X+V_1(x)\geq X+V_1(y)+V_1(-z)\in\cA+\cA^\infty\subset\cA$, so that $X+V_1(x)\in\cA$.
\end{proof}

\begin{theorem}
\label{thm:lsc_NoA_Llinear}
Assume one of the following sets of assumptions:
\begin{enumerate}
  \item[(i)] $\cA$ is convex and closed, $\cP$ is convex and closed, $V_0$ is convex and lower semicontinuous, and $V_1$ is superlinear and upper semicontinuous. Moreover, $\cL$ is a linear space.
  \item[(ii)] $\cA$ is closed, $\cP$ is closed, $V_0$ is lower semicontinuous, and $V_1$ is anti-star shaped and upper semicontinuous. Moreover, $\cL=\{0\}$.
  \item[(iii)] $\cA$ is closed, $\cP$ is compact, $V_0$ is lower semicontinuous, and $V_1$ is upper semicontinuous. Moreover, $\cL=\{0\}$.
\end{enumerate}
Then, $\cC$ is closed and $\rho$ is lower semicontinuous. Moreover, for every $X\in\cX$ with $\rho(X)\in\R$
\[
\rho(X)=\min\{V_0(x) \,; \ x\in\cP, \ X+V_1(x)\in\cA\}.
\]
\end{theorem}
\begin{proof}
Let $A$ be an index set for nets and take a net $((X_\alpha,m_\alpha))\subset\cC$ that converges to $(X,m)\in\cX\times\R$. For every $\alpha\in A$, there exists $x_\alpha\in\cP$ such that $X_\alpha+V_1(x_\alpha)\in\cA$ and $V_0(x_\alpha)\leq m_\alpha$. Under (i), we may assume without loss of generality that $(x_\alpha)\subset\cL^\perp$ by virtue of Lemma~\ref{lemma:wlog_lambda_N0rt}.

\smallskip

Under (iii), $(x_\alpha)$ has a convergent subnet by compactness. Suppose that either (i) or (ii) holds and $(x_\alpha)$ has no convergent subnets. In this case, we find a subnet of $(x_\alpha)$ consisting of nonzero elements with diverging norms. Indeed, it suffices to consider the index set $\{(\alpha,n)\in A\times\N \, : \ \|x_\alpha\|>n\}$, equipped with the direction defined by $(\alpha,n)\succeq(\beta,m)$ if and only if $\alpha\succeq\beta$ and $m\geq n$, and take $x_{(\alpha,n)}=x_\alpha$ for every $(\alpha,n)\in A\times\N$. For convenience, we still denote the diverging subnet by $(x_\alpha)$ and we may assume that $\norm{x_\alpha}\geq1$ for every $\alpha\in A$. If necessary by passing to a suitable subnet, we find a nonzero $x\in\R^N$ such that $\frac{x_\alpha}{\norm{x_\alpha}}\to x$. Note that $x\in\cP^\infty$. Under (i), we additionally have $x\in\cL^\perp$. Again passing to a suitable subnet if necessary, we find $(Y_\alpha)\subset\cX$ such that $Y_\alpha\to V_1(x)$ and for every $\alpha\in A$
\[
V_1\left(\frac{x_\alpha}{\norm{x_\alpha}}\right)\leq Y_\alpha
\]
by upper semicontinuity of $V_1$. Observe that for every $\alpha\in A$
\[
X_\alpha+\norm{x_{\alpha}}Y_{\alpha} \geq X_{\alpha}+{\norm{x_{\alpha}}}\,V_1\left( \frac{x_{\alpha}}{\norm{x_{\alpha}}}\right) \geq X_{\alpha}+{\norm{x_{\alpha}}}\frac{V_1(x_{\alpha})}{\norm{x_{\alpha}}} =  X_{\alpha}+V_1(x_{\alpha})\in\cA,
\]
showing that $X_\alpha+\norm{x_{\alpha}}Y_{\alpha}\in\cA$. Note also that
\[
\frac{X_\alpha+\norm{x_{\alpha}}Y_{\alpha}}{\norm{x_{\alpha}}}
=\frac{X_{\alpha}}{\norm{x_{\alpha}}}+Y_{\alpha} \xrightarrow[]{}V_1(x).
\]
As a result, $V_1(x)\in\cA^\infty$. Since $V_0(x_\alpha)\leq m_\alpha\leq m+1$ for each $\alpha\in A$, we additionally get
\[
x \in \{y\in\R^N \,; \ V_0(y)\leq m+1\}^\infty \subset \{y\in\R^N \,; \ V_0^\infty(y)\leq 0\}.
\]
This yields $x\in\cL$. We claim that $x=0$. This is clear under (ii) while it follows from $x\in\cL\cap\cL^\perp=\{0\}$ under (i). However, this conclusion cannot hold because $x$ is nonzero by definition.

\smallskip

The previous argument shows that $(x_\alpha)$ must have a convergent subnet. We may therefore assume without loss of generality that $x_\alpha\to x$ for some $x\in\cP$. Note that
\[
V_0(x) \leq \liminf\limits_{\alpha} V_0(x_\alpha)\leq \lim\limits_{\alpha} m_\alpha=m
\]
by lower semicontinuity of $V_0$. In addition, if necessary passing to a suitable subnet, we find $(Y_\alpha)\subset\cX$ such that $Y_\alpha\to V_1(x)$ and $Y_\alpha\geq V_1(x_\alpha)$ for every $\alpha\in A$ by upper semicontinuity of $V_1$. Clearly, $X_\alpha+Y_\alpha\geq X_\alpha+V_1(x_\alpha)\in\cA$, so that $X_\alpha+Y_\alpha\in\cA$, for every $\alpha\in A$. This implies that $X_\alpha+Y_\alpha\to X+V_1(x)\in\cA$. In conclusion, it follows that $(X,m)\in\cC$, showing that $\cC$ is closed. Lower semicontinuity of $\rho$ follows from the inclusions $\cC\subset\epi(\rho)\subset\cl(\cC)$ in Remark~\ref{rem: epigraph}. Now, take $X\in\cX$ and assume that $\rho(X)\in\R$. As $\cC$ is closed, we easily see that $\rho(X) = \min\{m\in\R \,; \ (X,m)\in\cC \}$ by Lemma~\ref{lem: defn C}. This yields $(X,\rho(X))\in\cC$ and implies that we find $x\in\cP$ such that $X+V_1(x)\in\cA$ and $V_0(x)\leq \rho(X)$. By definition of $\rho$, we also have $\rho(X)\leq V_0(x)$, proving the desired attainability and concluding the proof.
\end{proof}

\begin{remark}
(i) In the literature, it is standard to {\em assume} lower semicontinuity for a risk measure, e.g., to apply duality theory and establish dual representations. This is unproblematic in the usual frictionless single-asset setting where there is a clear correspondence between lower semicontinuity of $\rho$ and closedness of $\cA$. The assumption becomes more problematic in the presence of multiple eligible assets and/or market frictions as it is generally hard to trace lower semicontinuity back to the properties of the underlying financial primitives. To the best of our knowledge, the most general result on lower semicontinuity for a risk measure with multiple eligible assets is Proposition 5 in \cite{article:farkas2015measuring}. The preceding theorem extends that result beyond frictionless markets and convex/conic acceptance sets. In addition, we replace the absence of acceptable deals used there with assumptions on $\cL$ linked to the weaker absence of scalable acceptable deals; see Section~\ref{sect: acceptable deals}.

\smallskip

(ii) The proof of the preceding theorem was inspired by the argument in Theorem 8 in \cite{article:pennanen2011dual}, where lower semicontinuity of the superhedging price in a multi-period market with convex frictions is obtained as a consequence of the closedness of a set playing the role of our set $\cC$. The space $\cX$ is the space $L^0$ of random variables on a probability triple equipped with the topology of convergence in probability, and $\cA$ is the corresponding positive cone $L^0_+$. The proof, however, relies on pointwise arguments that cannot be applied to our general acceptance sets as they typically fail to be closed with respect to almost sure convergence.
\end{remark}


\section{Dual representations}
\label{sect: dual representations}


Dual representations have been widely investigated in the risk measure literature; see, e.g., \cite{article:artzner1999coherent}, \cite{article:follmer2002convex}, \cite{article:frittelli2002putting}, 
\cite{jouini2006law}, and \cite{cheridito2017duality} for the convex case and \cite{article:cerreia2011risk}, \cite{article:drapeau2013risk}, and \cite{gao2018fatou} for the quasiconvex case. In this section we establish dual representations of general risk measures under both convexity and quasiconvexity. The essence of any dual representation is to express a risk measure as a supremum of affine (in the convex case) or quasiaffine (in the quasiconvex case) functionals over a set of suitable dual elements that are interpreted as generalized pricing rules. As a key prerequisite to achieve such a representation, the risk measure has to be lower semicontinuous. We refer to Section~\ref{sect: lsc} for a thorough analysis of this property. Throughout this section we assume that $\cX$ is locally convex and denote by $\cX'$ its topological dual, which is partially ordered by the convex cone
\[
\cX'_+:=\{\psi\in\cX' \,; \ \psi(X)\geq0, \ \forall X\in\cX_+\}.
\]


\subsection{Convex risk measures}

We start by focusing on convex risk measures. In this case, it is useful to use the notation
\[
\sigma_\cA(\psi):=\inf_{X\in\cA}\psi(X), \ \ \ \sigma_{\cP,V_0,V_1}(\psi):=\inf_{x\in\cP}\{V_0(x)-\psi(V_1(x))\}, \ \ \ \psi\in\cX',
\]
\[
\cB:=\{\psi\in\cX' \,; \ \sigma_\cA(\psi)>-\infty\}, \ \ \ \cD := \{\psi\in\cX' \,; \ \sigma_\cA(\psi)>-\infty, \ \sigma_{\cP,V_0,V_1}(\psi)>-\infty\}.
\]
Observe that $\sigma_\cA$ is the (lower) support function of $\cA$ and the set $\cB$ is therefore the barrier cone of $\cA$. We refer to the appendix for the necessary details about support functions and barrier cones. The set $\cD$ will turn out to be the natural domain of our dual representation. Any element $\psi\in\cD$ can be interpreted as a pricing rule defined on the entire space $\cX$ that respects market prices and is consistent with the underlying acceptance set. On the one side, for every admissible portfolio $x\in\cP$
\[
\psi(V_1(x)) \leq V_0(x)-\sigma_{\cP,V_0,V_1}(\psi).
\]
This shows that the payoff $V_1(x)$ is ``priced'' through $\psi$ consistently with the ask price $V_0(x)$ up to a convenient adjustment. As demonstrated below, no adjustment is needed if market frictions are proportional. In fact, $\psi(V_1(x))$ actually coincides with $V_0(x)$ if there are no market frictions altogether. On the other side, for every acceptable position $X\in\cA$
\[
\psi(X) \geq \sigma_\cA(\psi).
\]
This shows that the ``price'' of acceptable positions is bounded from below, reflecting consistency with the underlying ``preferences'' of the agent embedded into the chosen notion of acceptability. As shown below, the bounding constant can be taken to be zero if the acceptance set is a cone. In this case, every acceptable position has a genuine ``price''.

\begin{proposition}
\label{prop: on D}
We have $\cD\subset\cB\subset\cX'_+$. Moreover, for $\psi\in\cD$ the following statements hold:
\begin{enumerate}
  \item[(i)] If $\cA$ is a cone, then $\psi(X)\geq0$ for every $X\in\cA$.
  \item[(ii)] If $\cP$ is a cone and $V_0$ and $V_1$ are positively homogeneous, then $\psi(V_1(x))\leq V_0(x)$ for every $x\in\cP$.
  \item[(iii)] If $\cP$ is linear and $V_0$ and $V_1$ are linear, then $\psi(V_1(x))=V_0(x)$ for every $x\in\cP$.
\end{enumerate}
\end{proposition}
\begin{proof}
Fix $\psi\in\cB$ and take $X\in\cX_+$. As $\cX_+\subset\cA$, it must hold that $\psi(\lambda X)\geq\sigma_\cA(\psi)>-\infty$ for every $\lambda\in\R_{++}$. This yields $\psi(X)\geq0$ and shows that $\cB\subset\cX'_+$. Now, let $\psi\in\cD$. Let $\cA$ be a cone and take $X\in\cA$. In the same vein, we must have $\psi(\lambda X)\geq\sigma_\cA(\psi)>-\infty$ for every $\lambda\in\R_{++}$, implying that $\psi(X)\geq0$. This delivers (i). Next, assume that $\cP$ is a cone and $V_0$ and $V_1$ are positively homogeneous. Then, for every $x\in\cP$ we must have
\[
\inf_{\lambda\in\R_{++}}\{\lambda[V_0(x)-\psi(V_1(x))]\} = \inf_{\lambda\in\R_{++}}\{V_0(\lambda x)-\psi(V_1(\lambda x))\} \geq \sigma_{\cP,V_0,V_1}(\psi) > -\infty.
\]
This can only hold if $V_0(x)-\psi(V_1(x))\geq0$, showing (ii). If additionally $\cP$ is linear and $V_0$ and $V_1$ are both linear, then we also obtain $\psi(V_1(x))-V_0(x)=V_0(-x)-\psi(V_1(-x))\geq0$, proving (iii).
\end{proof}

\begin{remark}
While in the conic case every $\psi\in\cD$ satisfies $\sigma_{\cP,V_0,V_1}(\psi)=0$, this is generally not true if one departs from conicity. In the convex case, we have the following result: If $\cP$ is convex, $V_0$ is convex, and $V_1$ is concave, then $\sigma_{\cP,V_0,V_1}(\psi)=0$ if and only if there exists a linear functional $\pi:\R^N\to\R$ such that $\psi(V_1(x))\leq \pi(x)\leq V_0(x)$ for every $x\in\cP$. This follows, e.g., from Theorem 4.3.2 in \cite{book:borwein2004techniques}. The functional $\pi$ defines frictionless ``shadow prices'' that are consistent with the market bid-ask spreads.
\end{remark}

The next example shows that, in the standard framework of random variables, the set $\cD$ is intimately linked to price deflators and martingale measures.

\begin{example}
\label{ex: martingale measure}
Let $(\Omega,\cF,\probp)$ be a probability space and take $\cX=L^1$ equipped with its standard norm topology. Then, every $\psi\in\cD$ can be represented by a suitable pricing density $D\in L^\infty$ as
\[
\psi(X)=\E_\probp[DX], \ \ \ X\in\cX.
\]
By the previous proposition, $\probp(D\geq0)=1$. If there exists $z\in\cP$ such that $\Span(\{z\})\subset\cP$ and such that $V_0$ and $V_1$ are linear along $z$ and satisfy $V_0(z)=1$ and $V_1(z)=1_\Omega$, then $\E_\probp[D]=1$ and we find a probability measure $\probq$ on $(\Omega,\cF)$ that is absolutely continuous with respect to $\probp$ and satisfies $\frac{d\probq}{d\probp}\in L^\infty$ and
\[
\psi(X)=\E_\probq[X], \ \ \ X\in\cX.
\]
Now, consider a frictionless setting where $\cP$ is linear and both $V_0$ and $V_1$ are linear on $\cP$. In this case,
\[
\E_\probp[DV_1(x)]=\E_\probq[X]=V_0(x), \ \ \ x\in\cP.
\]
This shows that $D$ is a price deflator and $\probq$ is a martingale measure.
\end{example}

The next result records the dual representation of convex risk measures. The penalty function in the general dual representation consists of two terms, namely the support functions $\sigma_\cA$ and $\sigma_{\cP,V_0,V_1}$. The appealing feature is that one can therefore disentangle the role of the acceptance set, embedded into $\sigma_\cA$, from that of the market model, embedded into $\sigma_{\cP,V_0,V_1}$, and adjust the representation as a consequence of the properties of the financial primitives. In particular, as shown at the beginning of this section, there are situations where the support functions are null on their domains. This occurs if the acceptance set is a cone and if the market has proportional frictions, respectively. In these situations, the dual representation simplifies considerably. Our result provides a unifying formulation for the dual representations known in the literature. In particular, it extends Theorem 3 in \cite{article:farkas2015measuring} beyond frictionless markets and convex acceptance sets, and Proposition 3.9 in \cite{article:frittelli2006risk} beyond ``efficient'' markets where two portfolios with the same liquidation price command the same acquisition price. As usual, we say that $\rho$ is proper if it is not identically $\infty$ and does not take the value $-\infty$.

\begin{theorem}
\label{theo: dual repr convex}
Let $\rho$ be proper, convex, and lower semicontinuous. Then, $\cD$ is nonempty and
\[
\rho(X)=\sup_{\psi\in\cD}\left\{\sigma_\cA(\psi)+\sigma_{\cP,V_0,V_1}(\psi)-\psi(X)\right\}, \ \ \ \ X\in\cX.
\]
\end{theorem}
\begin{proof}
We use standard conjugate duality, see, e.g., Theorem 2.3.3 in \cite{book:zalinescu2002convex}. A straightforward calculation of the convex conjugate of $\rho$ gives for every $\psi\in\cX'$
\begin{align*}
\rho^*(-\psi)
&:=
\sup_{X\in\cX}\{-\psi(X)-\rho(X)\} \\
&=
-\inf_{X\in\cX}\big\{\psi(X)+\inf\{V_0(x) \,; \ x\in\cP, \ X+V_1(x)\in\cA \}\big\} \\
&=
-\inf_{X\in\cX}\inf\{\psi(X)+V_0(x) \,; \ x\in\cP, \ X+V_1(x)\in\cA\}\\
&=
-\inf\{\psi(Z-V_1(x))+V_0(x) \,; \ x\in\cP, \ Z\in\cA\}\\
&=
-\inf\{\psi(Z)+V_0(x)-\psi(V_1(x)) \,; \ x\in\cP, \ Z\in\cA\}\\
&=
-\sigma_\cA(\psi)-\sigma_{\cP,V_0,V_1}(\psi).
\end{align*}
The Fenchel-Moreau representation of $\rho$ thus yields for $X\in\cX$
\[
\rho(X)=\sup_{\psi\in\cX'}\{-\psi(X)-\rho^*(-\psi)\}=
\sup_{\psi\in\cX'}\{\sigma_\cA(\psi)+\sigma_{\cP,V_0,V_1}(\psi)-\psi(X)\}.
\]
As $\rho$ is proper, the set $\cD$ must be nonempty and the supremum above can be restricted to $\cD$.
\end{proof}

The proof of the dual representation is based on standard convex duality. We now aim to take a step further and refine the dual representation by replacing the natural domain $\cD$ with a smaller domain. While this is a mathematically meaningful step, as it leads to a more parsimonious representation, our motivation is mainly driven by economics. Remember that the elements of $\cD$ can be interpreted as pricing rules in a ``shadow'' frictionless market where prices of portfolios of eligible assets are compatible with their original bid-ask spreads and the ``preferences'' of the agent are priced in in the sense that any pricing rule in $\cD$ is bounded from below when applied to acceptable positions. The last condition allows a pricing rule in $\cD$ to assign a nonpositive price to a nonzero acceptable position. From an economic perspective, it is natural to try to restrict the dual representation to only feature pricing rules assigning a strictly positive price to any nonzero ``desirable'' position. A first idea is to replace $\cD$ with the smaller domain
\[
\cD_{++} := \{\psi\in\cD \,; \ \psi(X)>0, \ \forall X\in\cX_+\setminus\{0\}\}.
\]
In this case, any nonzero positive position is viewed as desirable. The corresponding refined dual representations are well understood in a single-asset frictionless setting, see, e.g., Theorem 4.43 in \cite{book:follmer2011stochastic}. These results are mathematically challenging and are akin to risk-measure versions of the classical Superhedging Theorem, where the acceptance set plays the role of the standard positive cone, as they require moving to {\em strictly-positive} price deflators or, similarly, {\em equivalent} martingale measures; see Example~\ref{ex: martingale measure}. Our aim is to extend these representations to our general market with transaction costs and portfolio constraints. In fact, we argue that $\cD_{++}$ should be replaced by the even smaller domain
\[
\cD_{str} := \{\psi\in\cD \,; \ \psi(X)>0, \ \forall X\in\cA\setminus\{0\}\}.
\]
This stricter restriction is economically meaningful as the agent should be prepared to view any nonzero acceptable position as desirable. This is in the spirit of good deal pricing, see, e.g., \cite{article:cochrane2000beyond}, \cite{article:bernardo2000gain}, \cite{article:carr2001pricing}, \cite{article:jaschke2001coherent}, and requires moving to {\em special} strictly-positive price deflators or equivalent martingale measures. In the risk measure literature, a representation in terms of strict generalized pricing rules has been obtained, to the best of our knowledge, only in a frictionless setting where the model space consists of bounded random variables and $\cD_{str}=\cD_{++}$ as the acceptance set is taken to be the positive cone; see Proposition 4.99 in \cite{book:follmer2011stochastic}. We generalize this result by exploiting ideas from arbitrage pricing. To this end, we have to stipulate some additional assumptions on the underlying financial primitives. In particular, we will need the ``no arbitrage'' conditions encountered in Section~\ref{sect: acceptable deals}. Furthermore, our model spaces have to fulfill the following regularity properties.

\begin{definition}
We say that $(\cX,\cX')$ is an admissible pair if there exists a separable normed space $\cY$ such that $\cX'$ is the norm dual of $\cY$ and $\sigma(\cX',\cX)$ is weaker than $\sigma(\cX',\cY)$.
\end{definition}

\begin{remark}
Admissibility is a rather mild requirement in the framework of spaces of random variables. Let $(\Omega,\cF,\probp)$ be a probability space and take $\cX\subset L^1$ and $\cX'=L^\infty$. Moreover, set $\cY=L^1$. For every $X\in\cY$ define a linear functional on $\cX'$ by setting $\psi_X(Z)=\E_\probp[XZ]$. The topology $\sigma(\cX',\cX)$, respectively $\sigma(\cX',\cY)$, is the weakest topology on $\cX'$ making $\psi_X$ continuous for every $X\in\cX$, respectively $X\in\cY$. In particular, $\sigma(\cX',\cY)$ coincides with the classical weak-star topology on $\cX'$. It is immediate to see that $\sigma(\cX',\cX)$ is weaker than $\sigma(\cX',\cY)$. For $(\cX,\cX')$ to be admissible, the space $\cY$ must be separable. This is the case if, e.g., $\cF$ is countably generated. We refer to Theorem 13.16 in \cite{book:aliprantis1999infinite} for a characterization of the separability of $\cY$ in a nonatomic setting.
\end{remark}

The next theorem is the main result of this section and records the announced representation of risk measures by means of strict pricing rules. For convenience, we split the proof into several steps. We refer to the ensuing remark for the embedding in the literature.

\begin{theorem}
\label{theo: improved dual}
Let $\rho$ be proper, convex, and lower semicontinuous. Moreover, let $(\cX,\cX')$ be an admissible pair, and assume that $\cA$ is a closed convex cone with $\cA\cap(-\cA)=\{0\}$, $\cP$ is closed, $V_0$ is lower semicontinuous, and $V_1$ is anti-star shaped and upper semicontinuous. Furthermore, assume that $\cL=\{0\}$. Then, $\cD_{str}$ is nonempty and $\sigma_\cA(\psi)=0$ for every $\psi\in\cD_{str}$. Moreover,
\[
\rho(X)=\sup_{\psi\in\cD_{str}}\left\{\sigma_{\cP,V_0,V_1}(\psi)-\psi(X)\right\}, \ \ \ \ X\in\cX.
\]
\end{theorem}
\begin{proof}
Throughout the proof, we denote by $\cY$ a separable normed space such that $\cX'$ is the norm dual of $\cY$ and $\sigma(\cX',\cX)$ is weaker than $\sigma(\cX',\cY)$. Moreover, we denote by $\|\cdot\|_{\cX'}$ the corresponding dual norm on $\cX'$. The existence of $\cY$ follows from admissibility of $(\cX,\cX')$.

\smallskip

{\em Step 1}. First, we show that the set $\cP_0=\{x\in\cP \,; \ V_0(x)\leq0, \ V_1(x)\in\cA\}$ is bounded. By contradiction, take a sequence $(x_n)\subset\cP_0$ such that $\norm{x_n}\geq n$ for every $n\in\N$. By passing to a convenient subsequence if necessary, we find a nonzero $x\in\R^N$ such that $\frac{x_n}{\norm{x_n}}\to x$. Note that $x\in\cP^\infty$. Moreover,
\[
x \in \{y\in\R^N \,; \ V_0(y)\leq0\}^\infty \subset \{y\in\R^N \,; \ V_0^\infty(y)\leq0\}.
\]
As $V_1$ is upper semicontinuous, we find a subnet $\big(\frac{x_\alpha}{\norm{x_\alpha}}\big)$ and a net $(Y_\alpha)\subset\cX$ such that $V_1\big(\frac{x_\alpha}{\norm{x_\alpha}}\big)\leq Y_\alpha$ for every $\alpha$ and $Y_\alpha\to V_1(x)$. It follows that for every $\alpha$
\begin{equation*}
\norm{x_\alpha}Y_{\alpha} \geq {\norm{x_\alpha}}\,V_1\Big( \tfrac{x_\alpha}{\norm{x_\alpha}} \Big)	\geq \tfrac{\norm{x_\alpha}}{\norm{x_\alpha}}V_1(x_\alpha) = V_1(x_\alpha)\in\cA.
\end{equation*}
As a result, we obtain $\norm{x_\alpha}Y_{\alpha}\in\cA$ for every $\alpha$. Since we have
\begin{equation*}
\tfrac{\norm{x_\alpha}Y_{\alpha}}{\norm{x_\alpha}} = Y_{\alpha} \xrightarrow[]{}V_1(x),
\end{equation*}
we infer that $V_1(x)\in\cA^\infty$. However, this would imply that $x\in\cL=\{0\}$, which cannot hold because $x$ was nonzero by definition. In conclusion, $\cP_0$ must be bounded as claimed.

\smallskip

{\em Step 2}. Take a nonzero $X\in\cA$. We claim that $(-\lambda X,0)\notin\cC$ for some $\lambda\in(1,\infty)$. By contradiction, assume that for every $\lambda\in(1,\infty)$ there exists $x_\lambda\in\cP$ such that $V_0(x_\lambda)\leq0$ and $-\lambda X+V_1(x_\lambda)\in\cA$. We claim that $(x_\lambda)\subset\cP_0$. To see this, it suffices to observe that $V_1(x_\lambda)\in\cA+\lambda X\subset\cA$ for every $\lambda\in(1,\infty)$ by assumption on $\cA$. As a consequence, the net $(x_\lambda)$ must be bounded by Step 1. By passing to a convenient subnet if necessary, we may assume that $x_\lambda\to x$ for some $x\in\R^N$. In particular, $\tfrac{x_\lambda}{\lambda}\to0$. For $\lambda\in(1,\infty)$
\[
V_1\left(\tfrac{x_\lambda}{\lambda}\right) \geq \tfrac{1}{\lambda}V_1(x_\lambda) \in \tfrac{1}{\lambda}(\cA+\lambda X) \subset \cA+X,
\]
showing that $V_1\left(\tfrac{x_\lambda}{\lambda}\right)\in\cA+X$. By upper semicontinuity of $V_1$ and again by passing to a suitable subnet if necessary, we find $(Z_\lambda)\subset\cX$ such that $Z_\lambda\geq V_1\left(\tfrac{x_\lambda}{\lambda}\right)$ for every $\lambda$ and $Z_\lambda\to V_1(0)=0$. Note that $(Z_\lambda)\subset\cA+X$ and, hence, $0\in\cA+X$ by closedness of $\cA$. However, this cannot hold because $\cA\cap(-\cA)=\{0\}$. This implies that $(-\lambda X,0)\notin\cC$ for some $\lambda\in(1,\infty)$ and concludes the proof.

\smallskip

{\em Step 3}. Let $\sigma_\cC$ be the lower support function of $\cC$. We show that every $\psi\in\cX'$ such that $\sigma_\cC(\psi,1)>-\infty$ belongs to $\cD$. To this effect, observe that
\begin{align*}
\sigma_\cC(\psi,1)
&=
\inf\{\psi(X)+m \,; \ (X,m)\in\cX\times\R, \ \exists x\in\cP, \ V_0(x)\leq m, \ X+V_1(x)\in\cA\} \\
&=
\inf\{\psi(Y-V_1(x))+V_0(x) \,; \ Y\in\cA, \ x\in\cP\} \\
&=
\inf\{\psi(Y)+V_0(x)-\psi(V_1(x)) \,; \ Y\in\cA, \ x\in\cP\} \\
&=
\sigma_\cA(\psi)+\sigma_{\cP,V_0,V_1}(\psi).
\end{align*}
This implies that $\sigma_\cA(\psi)>-\infty$ and $\sigma_{\cP,V_0,V_1}(\psi)>-\infty$, and delivers the desired claim.

\smallskip

{\em Step 4}. We show that $(0,-\lambda)\notin\cC$ for some $\lambda\in\R_{++}$. To the contrary, suppose that for every $\lambda\in\R_{++}$ we find $x_\lambda\in\cP$ such that $V_0(x_\lambda)\leq-\lambda$ and $V_1(x_\lambda)\in\cA$. The net $(x_\lambda)$ is bounded by Step 1. Then, we may assume without loss of generality that $x_\lambda\to x$ for some $x\in\R^N$. This implies
\[
V_0(x) \leq \liminf_\lambda V_0(x_\lambda) = -\infty
\]
by lower semicontinuity of $V_0$. As this contradicts our assumptions on $V_0$, there must exist $\lambda\in\R_{++}$ such that $(0,-\lambda)\notin\cC$, proving the desired claim.

\smallskip

{\em Step 5}. Let $X\in\cX$. We show that, if $(X,0)\notin\cC$, then there exists $\psi\in\cD$ such that $\psi(X)<0$. First, note that $\cC$ is closed by Theorem~\ref{thm:lsc_NoA_Llinear}. Furthermore, as $\rho$ is convex and $\cC$ coincides with its epigraph by Remark~\ref{rem: epigraph}, we infer that $\cC$ is convex. A direct application of the Hahn-Banach theorem, see, e.g., Theorem 1.1.9 in \cite{book:zalinescu2002convex}, implies that $\varphi(X)<\sigma_\cC(\varphi,a)\leq0$ for some $\varphi\in\cX'$ and $a\in\R$ with $\sigma_\cC(\varphi,a)>-\infty$. In particular, $\varphi(X)<0$. Note that $(0,\lambda)\in\cC$ for every $\lambda\in\R_{++}$, so that
\[
\inf_{\lambda\in\R_{++}}\{\lambda a\} \geq \sigma_\cC(\varphi,a) > -\infty.
\]
This yields $a\geq0$. Since $(0,-\lambda)\notin\cC$ for some $\lambda\in\R_{++}$ by Step 4, we must find as above $\chi\in\cX'$ and $b\in\R$ such that $\sigma_\cC(\chi,b)>-\infty$ and $-\lambda b<\sigma_\cC(\chi,b)\leq0$. Hence, $b>0$. For every $\lambda\in(0,1)$ define $\psi_\lambda=\lambda\varphi+(1-\lambda)\chi\in\cX'$ and $c_\lambda=\lambda a+(1-\lambda)b\in\R_{++}$. To conclude the proof, take $\lambda$ close enough to $1$ to have $\psi_\lambda(X)<0$ and observe that
\[
\sigma_\cC\Big(\tfrac{\psi_\lambda}{c_\lambda},1\Big) = \tfrac{1}{c_\lambda}\sigma_\cC(\psi_\lambda,c_\lambda) \geq \tfrac{\lambda}{c_\lambda}\sigma_\cC(\varphi,a)+\tfrac{1-\lambda}{c_\lambda}\sigma_\cC(\chi,b) > -\infty
\]
by concavity of $\sigma_\cC$, implying that $\frac{\psi_\lambda}{c_\lambda}\in\cD$ by Step 3. It now suffices to define $\psi=\frac{\psi_\lambda}{c_\lambda}$.

\smallskip

{\em Step 6}. We claim that for every sequence $(\psi_n)\subset\cD$ there exist a sequence $(\lambda_n)\subset\R_{++}$ and $\psi\in\cD$ such that $\sum_{k=1}^n \lambda_k\psi_k\to \psi$ in the topology $\sigma(\cX',\cX)$. To prove this, note first that $\sigma_\cC(\psi,1)\leq0$ for every $\psi\in\cD$. For every $n\in\N$ set $a_n=(1+\|\psi_n\|_{\cX'})^{-1}(1-\sigma_\cC(\psi_n,1))^{-1}2^{-n}>0$ and define $\varphi_n=\sum_{k=1}^na_k\psi_k\in\cX'$. Note that $(\varphi_n)$ is a Cauchy sequence in the norm topology of $\cX'$. Since $\cX'$ is complete, we have $\varphi_n\to \varphi$ for a suitable $\varphi\in\cX'$ with respect to its norm topology. A fortiori, $\varphi_n\to\varphi$ with respect to the weak-star topology $\sigma(\cX',\cY)$. As $\sigma(\cX',\cX)$ is weaker than $\sigma(\cX',\cY)$ by assumption, we also have convergence with respect to $\sigma(\cX',\cX)$. To conclude the proof, note that $\sum_{k=1}^na_k\to a$ for some $a\in\R_{++}$ and
\[
\sigma_\cC(\varphi,a) \geq \limsup_{n\to\infty}\sum_{k=1}^na_k\sigma_\cC(\psi_k,1) > -\infty
\]
by upper semicontinuity and superlinearity of $\sigma_\cC$. The desired claim follows by setting $\lambda_n=\frac{a_n}{a}\in\R_{++}$ for every $n\in\N$ and $\psi=\frac{\varphi}{a}\in\cX'$. Indeed, it is clear that $\sum_{k=1}^n \lambda_k\psi_k\to \psi$ with respect to $\sigma(\cX',\cX)$. In addition, we have $\sigma_\cC(\psi,1)=\frac{1}{a}\sigma_\cC(\varphi,a)>-\infty$, hence $\psi\in\cD$ by Step 3.

\smallskip

{\em Step 7}. We exhibit a sequence $(\psi_n)\subset\cD$ such that for every nonzero $X\in\cA$ there exists $n\in\N$ with $\psi_n(X)>0$. To this effect, take a nonzero $X\in\cA$. By Step 2, there exists $\lambda\in\R_{++}$ such that $(-\lambda X,0)\notin\cC$. Then, by Step 5, we find $\psi_X\in\cD$ satisfying $\psi_X(X)>0$. Now, consider the rescaled couple
\[
(\varphi_X,r_X) = \left(\tfrac{\psi_X}{\|\psi_X\|_{\cX'}},\tfrac{1}{\|\psi_X\|_{\cX'}}\right).
\]
As $\cY$ is separable by assumption, the unit ball in $\cX'$ is $\sigma(\cX',\cY)$-metrizable by Theorem 6.30 in \cite{book:aliprantis1999infinite} and $\sigma(\cX',\cY)$-compact by the Banach-Alaoglu theorem, see, e.g., Theorem 6.21 in \cite{book:aliprantis1999infinite}. As a consequence, the unit ball together with any of its subsets is $\sigma(\cX',\cY)$-separable. In particular, this is true of $\{\varphi_X \,; \ X\in\cA\setminus\{0\}\}$. Let $\{\varphi_{X_n} \,; \ n\in\N\}$ be a countable $\sigma(\cX',\cY)$-dense subset. Since convergence in $\sigma(\cX',\cY)$ implies convergence in $\sigma(\cX',\cX)$, for every nonzero $X\in\cA$ we must have $\varphi_{X_n}(X)>0$, hence $\psi_{X_n}(X)>0$, for some $n\in\N$ by density. This delivers the desired assertion.

\smallskip

{\em Step 8}. We establish the assertion of the theorem. Take the sequence $(\psi_n)\subset\cD$ from Step 7. It follows from Step 6 that we find $(\lambda_n)\subset\R_{++}$ and $\psi\in\cD$ such that $\sum_{k=1}^n \lambda_k\psi_k\to \psi$ in the topology $\sigma(\cX',\cX)$. Take an arbitrary nonzero $X\in\cA$. For every $n\in\N$ we have $\psi_n(X)\geq0$ by conicity of $\cA$, see Proposition~\ref{prop: on D}. In addition, we find $n\in\N$ such that $\psi_n(X)>0$ by the defining properties of our sequence. It is therefore clear that $\psi(X)>0$, showing that $\psi\in\cD_{str}$. Now, take $X\in\cX$. It follows from Theorem~\ref{theo: dual repr convex} that
\[
\rho(X) = \sup_{\psi\in\cD}\{\sigma_{\cP,V_0,V_1}(\psi)-\psi(X)\},
\]
where we used that $\sigma_\cA(\psi)=0$ for every $\psi\in\cD$ by Proposition~\ref{prop: on D}. It remains to prove that
\[
\sup_{\psi\in\cD}\{\sigma_{\cP,V_0,V_1}(\psi)-\psi(X)\} \leq \sup_{\psi\in\cD_{str}}\{\sigma_{\cP,V_0,V_1}(\psi)-\psi(X)\}.
\]
To this end, let $\varphi\in\cD$ and $\psi\in\cD_{str}$. For every $\lambda\in(0,1)$ define $\psi_\lambda=\lambda\varphi+(1-\lambda)\psi\in\cD_{str}$ and note that
\[
\sigma_{\cP,V_0,V_1}(\psi_\lambda)-\psi_\lambda(X) \geq \lambda(\sigma_{\cP,V_0,V_1}(\varphi)-\varphi(X))+(1-\lambda)(\sigma_{\cP,V_0,V_1}(\psi)-\psi(X)) \to \sigma_{\cP,V_0,V_1}(\varphi)-\varphi(X)
\]
as $\lambda\uparrow1$. This yields the desired inequality and concludes the proof.
\end{proof}

\begin{remark}
(i) The last proof exploits ideas from arbitrage pricing. Indeed, the combination of Steps 6 and 7 used at the beginning of the proof of Step 8 is reminiscent of the arguments behind the Kreps-Yan theorem (see \cite{article:yan1980caracterisation} and \cite{article:kreps1981arbitrage}) which underpins some classical proofs of the Fundamental Theorem of Asset Pricing. The original exhaustion argument in \cite{article:yan1980caracterisation} cannot be reproduced in our general setting because the acceptance set may contain nonpositive positions whereas the arguments in \cite{article:kreps1981arbitrage} can be adapted beyond the frictionless setting of that paper. In so doing, we had to cope with the presence of market frictions and with the absence of convexity and we work under the assumption $\cL=\{0\}$ instead of viability as in that paper. A similar idea has been used in \cite{article:arduca2020market} in the presence of convexity and with a focus on payoffs instead of portfolios.

\smallskip

(ii) It is natural to ask if the same representation holds under weaker assumptions on the acceptance set. The pointedness condition $\cA\cap(-\cA)=\{0\}$ is clearly necessary for $\cD_{str}$ to be nonempty. Conicity of $\cA$ is necessary for Step 2 in our proof. To see this, let $\cX$ be the space of random variables on the probability space $(\Omega,\cF,\probp)$, where $\Omega=\{\omega_1,\omega_2\}$, $\cF$ coincides with the parts of $\Omega$, and $\probp$ assigns probability $1/2$ to each scenario. Set $\cP=\R^2$ and define
\[
\cA=\{X\in\cX \,; \ X(\omega_1)\geq0, \ X(\omega_2)\geq0\}\cup\{X\in\cX \,; \ -1\leq X(\omega_1)\leq0, \ X(\omega_1)+X(\omega_2)\geq0\}.
\]
Moreover, define the pricing rules by
\[
V_0(x)=x_1+x_2, \ \ \ V_1(x)=(x_1,x_2), \ \ \ x\in\R^2.
\]
It is not difficult to verify that the assumptions of our theorem hold apart from conicity of $\cA$. In particular, $\cL=\{0\}$ holds because $\cA^\infty=\{X\in\cX \,; \ X(\omega_1)\geq0, \ X(\omega_2)\geq0\}$. Now, let $X=(-1,1)\in\cA$ and take any $\lambda\in\R_{++}$. Observe that $x_\lambda=(-\lambda,\lambda)\in\cP$ satisfies $V_0(x_\lambda)=0$ and $-\lambda X+V_1(x_\lambda)=0\in\cA$. This shows that $(-\lambda X,0)\in\cC$. This shows that Step 2 does not hold in this case. A statement for general acceptance sets can be obtained if the assumptions of the theorem are assumed to hold for the enlarged acceptance set $\cK(\cA)+\cX_+$, where $\cK(\cA)$ is the closure of the convex cone generated by $\cA$.
\end{remark}

\subsection{Quasiconvex risk measures}

We turn to the quasiconvex case. Throughout this section, we define for $X\in\cX$ and $\psi\in\cX'$
\[
\rho(X|\psi):=\inf\{\rho(Y) \,; \ Y\in\cX, \ \psi(Y)\leq \psi(X) \}.
\]
This quantity corresponds to the smallest level of required capital attached to financial positions that are less expensive, from the perspective of the pricing rule $\psi$, than $X$. The next lemma records a useful representation of such quantity, showing that the functional $\rho(\cdot|\psi)$ corresponds to a risk measure with the same financial primitives but with enlarged acceptance set
\[
\cA_\psi:=\{X\in\cX \,; \ \exists Y\in\cX, \ \psi(Y)\leq0, \ X+Y\in\cA\}.
\]
This set consists of all positions that can be made acceptable ``at zero cost'', i.e., upon aggregation with a position having nonnegative price according to $\psi$.

\begin{lemma}
\label{lemma: rho(X|psi)}
For all $X\in\cX$ and $\psi\in\cX'$ we have
\[
\rho(X|\psi)=\inf\big\{V_0(x) \,; \ x\in\cP, \ X+V_1(x)\in\cA_\psi\big\}.
\]
\end{lemma}
\begin{proof}
It is easy to see that
    \begin{align}
	\rho(X |\psi) &= \inf\{\rho(Y) \,; \ Y\in\cX, \ \psi(Y-X)\leq 0\}\nonumber\\
	&= \inf\{\rho(X+Z) \,; \ Z\in\cX, \ \psi(Z)\leq 0\}\nonumber\\
	&= \inf\left\{ \inf\{V_0(x) \,; \ x\in\cP, \ X+Z+V_1(x)\in\cA \} \,; \ Z\in\cX, \ \psi(Z)\leq 0 \right\}\nonumber\\
	&= \inf\{ V_0(x) \,; \ x\in\cP, \ Z\in\cX, \ \psi(Z)\leq 0, \ X+Z+V_1(x)\in\cA \}\nonumber\\
	&=\inf\{V_0(x) \,; \ x\in\cP, \ X+V_1(x)\in\cA_\psi\}.\nonumber\qedhere
    \end{align}
\end{proof}

We can now establish the desired dual representation under quasiconvexity.

\begin{theorem}
\label{theo: dual representation quasiconvex}
Let $\rho$ be proper, quasiconvex, and lower semicontinuous. Then, $\cB$ is nonempty and
\[
\rho(X) = \sup_{\psi\in\cB}\rho(X|\psi), \ \ \ X\in\cX.
\]
\end{theorem}
\begin{proof}
It follows from the general dual representation of quasiconvex functions recorded in Theorem 3.8 in \cite{article:penot1990quasi} that for every $X\in\cX$ we can rewrite $\rho$ as
\begin{equation}
\label{eq: quasiconvex}
\rho(X)=\sup_{\psi\in\cX'}\rho(X|\psi).
\end{equation}
Now, take $\psi\in\cX'\setminus\cB$. Clearly, for every $Y\in\cX$ we find $Z\in\cA$ for which $\psi(Z-Y)\leq0$ holds, so that $Y\in\cA_\psi$. Thus, $\cA_\psi=\cX$ and, by virtue of Lemma~\ref{lemma: rho(X|psi)}, we have $\rho(X|\psi)=\inf\{V_0(x) \,; \ x\in\cP\}$. Note also that $\rho(X|0)=\rho(X)$ again by Lemma~\ref{lemma: rho(X|psi)}. As a result, $\rho(X|\psi)\leq\rho(X|0)$. Since $0\in\cB$, the supremum in~\ref{eq: quasiconvex} can be restricted to $\cB$. This concludes the proof.
\end{proof}

The next example shows that the domain of the dual representation of quasiconvex risk measures cannot generally be restricted further as in the convex case. On the one hand, we show that $\cB$ cannot be replaced with either $\cD_{str}$, even under the assumptions of Theorem~\ref{theo: improved dual}, or $\cD$ as in Theorem~\ref{theo: dual repr convex}. On the other hand, $\cB$ cannot be generally replaced with the smaller domain
\[
\cB_{str}:=\{\psi\in\cB \,; \ \psi(X)>0, \ \forall X\in\cA\setminus\{0\}\}.
\]
These observations seem to be new already in a frictionless setting.

\begin{example}
Let $\cX$ be the space of random variables on the probability space $(\Omega,\cF,\probp)$ where $\Omega=\{\omega_1,\omega_2\}$, $\cF$ coincides with the parts of $\Omega$, and $\probp$ assigns probability $1/2$ to each scenario. Set $\cP=\R$ and define $\cA=\{X\in\cX \,; \ X(\omega_1)\geq0, \ X(\omega_2)\geq0\}$ and
\[
V_0(x)=\begin{cases}
x+\tfrac{1}{2} & \mbox{if} \ x<-1,\\
\tfrac{1}{2}x & \mbox{if} \ -1\leq x<0,\\
x & \mbox{if} \ x\geq0,\\
\end{cases}
 \ \ \ V_1(x)=xS, \ \ \ x\in\R,
\]
where $S=1_\Omega\in\cX$. The assumptions on the financial primitives stipulated in Theorem~\ref{theo: improved dual} are all satisfied. In particular, $\cL = \{x\in\R_+ \,; \ V^\infty_0(x)\leq 0\}=\{0\}$. 
Note that for every $X\in\cX$
\[
\rho(X) = \inf\{V_0(x) \,; \ x\in\R, \ x\geq\max\{-X(\omega_1),-X(\omega_2)\}\} = V_0(\max\{-X(\omega_1),-X(\omega_2)\}).
\]
This shows that $\rho$ is proper, quasiconvex, and lower semicontinuous. Quasiconvexity and lower semicontinuity also follow from Proposition~\ref{prop: properties of rho} and Theorem~\ref{thm:lsc_NoA_Llinear}. Note that
\[
\cB=\{\psi\in\cX' \,; \ \psi(1_{\{\omega_1\}})\geq0, \ \psi(1_{\{\omega_2\}})\geq0\}, \ \ \ 
\cB_{str}=\{\psi\in\cX' \,; \ \psi(1_{\{\omega_1\}})>0, \ \psi(1_{\{\omega_2\}})>0\}.
\]
For every $\psi\in\cX'$ it is easy to see that
\[
\inf_{x\geq0}\{V_0(x)-\psi(V_1(x))\}>-\infty \ \iff \ \inf_{x>0}\{x-\psi(S)x\}>-\infty \ \iff \ \psi(S)\leq1,
\]
\[
\inf_{x<0}\{V_0(x)-\psi(V_1(x))\}>-\infty \ \iff \ \inf_{x<-1}\{x-\psi(S)x\}>-\infty \ \iff \ \psi(S)\geq1.
\]
As a result, we infer from Proposition~\ref{prop: on D} that
\[
\cD=\{\psi\in\cB \,; \ \psi(S)=1\}, \ \ \ \cD_{str}=\{\psi\in\cB_{str} \,; \ \psi(S)=1\}.
\]
Note that for every nonzero $\psi\in\cB$ we have $\cA_\psi=\cX$. Indeed, for $X\in\cX$ define
\[
Y=\begin{cases}
\min\left\{X(\omega_1),
-\tfrac{\psi(1_{\{\omega_2\}})}{\psi(1_{\{\omega_1\}})}X(\omega_2)\right\}1_{\{\omega_1\}}+X(\omega_2)1_{\{\omega_2\}} & \mbox{if} \ \psi(1_{\{\omega_1\}})>0,\\
X(\omega_1)1_{\{\omega_1\}}+\min\left\{X(\omega_2),0\right\}1_{\{\omega_2\}} & \mbox{if} \ \psi(1_{\{\omega_1\}})=0,
\end{cases}
\]
and observe that $X-Y\in\cA$ and $\psi(Y)\leq 0$. Since $\cA_\psi=\cX$, we infer from Lemma~\ref{lemma: rho(X|psi)} that $\rho(X|\psi)=-\infty$ for every $X\in\cX$. In addition, note that $0\notin\cD$ because $\sigma_{\cP,V_0,V_1}(0)=\inf\{V_0(x) \,; \ x\in\cP\}=-\infty$. In conclusion, for every $X\in\cX$ we obtain
\[
\rho(X)\neq\sup_{\psi\in\cB_{str}}\rho(X|\psi)=\sup_{\psi\in\cD}\rho(X|\psi)=\sup_{\psi\in\cD_{str}}\rho(X|\psi)=-\infty.
\]
\end{example}


\appendix


{\small

\section{Appendix}
\label{sect: math background}

In the appendix we review the key mathematical notions used in the paper with special emphasis on recession cones and functions. For more details we refer the reader to \cite{book:zalinescu2002convex}.

\smallskip

We adopt the convention $\infty-\infty=-\infty$ and $0\cdot\infty=0$, and define $\R_+:=[0,\infty)$ and $\R_{++}:=(0,\infty)$. Let $\cX$ be a real topological vector space. A nonempty set $\cC\subset\cX$ is star shaped if $\lambda X\in\cC$ for all $\lambda\in[0,1]$ and $X\in\cC$, a cone if $\lambda X\in\cC$ for all $\lambda\in\R_+$ and $X\in\cC$, convex if $\lambda X+(1-\lambda)Y\in\cC$ for all $\lambda\in[0,1]$ and $X,Y\in\cC$, closed under addition if $X+Y\in\cC$ for all $X,Y\in\cC$. Let $\cX$ be a real topological vector space and take a nonempty set $\cC\subset\cX$. The asymptotic cone of $\cC$ is
\[
\cC^\infty := \{X\in\cX \,; \ \exists (X_\alpha)\subset\cC, \ (\lambda_\alpha)\subset\R_+, \ \lambda_\alpha\to0, \ \lambda_\alpha X_\alpha\to X\}.
\]
Note that $\cC^\infty$ is always a closed cone, which coincides with the closure of $\cC$ when $\cC$ is itself a cone. If $\cC$ is closed and star shaped, then $\cC^\infty\subset\cC$. If $\cC$ is convex, then $\cC^\infty$ is also convex. If $\cC$ is convex and closed, for every $X\in\cC$
\[
\cC^\infty = \bigcap_{\lambda\in\R_{++}}\lambda(\cC-X\big).
\]
In this case, $\cC+\cC^\infty\subset\cC$. Now, let $\cY$ be another real topological vector space equipped with a partial order. A function $f:\cX\to\cY$ is star shaped if $f(\lambda X)\leq\lambda f(X)$ for all $\lambda\in[0,1]$ and $X\in\cX$, positively homogeneous if $f(\lambda X)=\lambda f(X)$ for all $\lambda\in\R_+$ and $X\in\cX$, convex if $f(\lambda X+(1-\lambda)Y)\leq\lambda f(X)+(1-\lambda)f(Y)$ for all $\lambda\in[0,1]$ and $X,Y\in\cX$, subadditive if $f(X+Y)\leq f(X)+f(Y)$ for all $X,Y\in\cX$. If $\cY$ is equipped with a lattice, then $f$ is called quasiconvex if $f(\lambda X+(1-\lambda)Y)\leq\max\{f(X),f(Y)\}$ for all $\lambda\in[0,1]$ and $X,Y\in\cX$.
We say that $f$ is anti-star shaped, concave, quasiconcave, superadditive whenever $-f$ is star shaped, convex, quasiconvex, subadditive. Now, consider a function $f:\cX\to[-\infty,\infty]$. We say that $f$ is proper if $f$ takes some finite value and never takes the value $-\infty$. We say that $f$ is lower semicontinuous if for all nets $(X_\alpha)\subset\cX$ and $X\in\cX$ we have $f(X)\leq\liminf f(X_\alpha)$ as $X_\alpha\to X$. The asymptotic function of $f$ is the unique function $f^\infty:\cX\to[-\infty,\infty]$ such that $\epi(f^\infty)=\epi(f)^\infty$, where $\epi(f) := \{(X,m)\in\cX\times\R \,; \ f(X)\leq m\}$ is the epigraph of $f$. 
The function $f^\infty$ is always lower semicontinuous and we either have $f^\infty(0)=0$ or $f^\infty(0)=-\infty$. In the former case, $f^\infty$ is positively homogeneous. If $f$ is lower semicontinuous and star shaped, then $f^\infty\leq f$. If $f$ is convex, then $f^\infty$ is also convex. If $f$ is proper, convex, and lower semicontinuous, then
\begin{equation*}
f^\infty(X) = \sup_{\lambda\in\R_{++}}\tfrac{f(Z+\lambda X)-f(Z)}{\lambda}, \ \ \ X\in\cX,
\end{equation*}
for every $Z\in\cX$ with $f(Z)\in\R$.

}


{\small

\bibliographystyle{apalike}
\bibliography{bibliography}

}


\end{document}